\documentclass[12pt]{article}
\usepackage{amsfonts}
\usepackage{amsmath}
\usepackage{amssymb}
\usepackage{natbib}
\usepackage{dsfont}
\usepackage{bbold}
\usepackage{subcaption}
\usepackage{setspace}
\usepackage{geometry}
\usepackage{graphicx}
\usepackage{secfix}

\newtheorem{theorem}{Theorem}[section]

\newtheorem{algorithm}[theorem]{Algorithm}

\newtheorem{condition}[theorem]{Assumption}

\newtheorem{corollary}[theorem]{Corollary}

\newtheorem{definition}[theorem]{Definition}

\newtheorem{lemma}[theorem]{Lemma}

\newtheorem{proposition}[theorem]{Proposition}

\newenvironment{proof}[1][Proof]{\noindent\textbf{#1.} }{\hfill \rule{0.5em}{0.5em}}
\newcommand{\fixqed}{\vspace*{-6ex}}

\DeclareMathOperator{\eigval}{eigval}
\DeclareMathOperator{\dimop}{dim}

\begin{document}

\title{Optimally-Transported\\Generalized Method of Moments}
\author{Susanne Schennach\thanks{%
Support from NSF grants SES-1950969 and SES-2150003 is gratefully acknowledged. The authors would like to thank Alfred Galichon, Florian Gunsilius, Heejun Lee and seminar participants at NYU's Advanced mathematical modeling in economics seminar, at the Econometrics and Optimal Transport Workshop at the University of Washington, at the Women in Econometrics Conference at the University of Toronto, at the Workshop on Optimal Transport in Econometrics at Collegio Carlo Alberto, and at the Aarhus Workshop in Econometrics VII for useful comments.}~\thanks{%
smschenn@brown.edu}~ and Vincent Starck\thanks{Financial support from the European Research Council (Starting Grant No. 852332) is gratefully acknowledged.}~\thanks{%
V.Starck@lmu.de} \\
Brown University, LMU München}
\maketitle

\begin{abstract}
%

We propose a novel optimal transport-based version of the Generalized
Method of Moment (GMM).
Instead of handling overidentification by reweighting
the data to satisfy the moment conditions (as in Generalized
Empirical Likelihood methods), this method proceeds by allowing for
errors in the variables of the least mean-square magnitude necessary to
simultaneously satisfy all moment conditions. This approach, based on the
notions of optimal transport and Wasserstein metric, aims to address the problem of assigning a
logical interpretation to GMM results even when overidentification tests
reject the null, a situation that cannot always be avoided in applications.
We illustrate the method by revisiting Duranton, Morrow and Turner's (2014) study of the relationship between a city's exports and the extent of its transportation infrastructure.
Our results corroborate theirs under weaker assumptions and provide
insight into the error structure of the variables. 

\medskip


\end{abstract}%

\section{Introduction}

The Generalized Method of Moment (GMM) (\citet{hansen:gmm}) has long been
the workhorse of statistical modeling in economics and the social sciences.
Its key distinguishing feature, relative to the basic method of moments, is
the presence of overidentifying restrictions that enable the model's
validity to be tested (\citet{Newey:HB}). With this ability to test comes
the obvious practical question of what one should do if an overidentified
GMM model fails overidentification tests, a situation that is not uncommon
(as noted in \citet{hall:missgmm}, \citet{hansen:miss}, %
\citet{poirier:salvaging}, \citet{conley:plausexo}, \citet{kwon:ineqmiss}),
even for perfectly reasonable, economically grounded, models.

A popular approach has been to find the \textquotedblleft
pseudo-true\textquotedblright\ value of the model parameter (%
\citet{sawa:pseudo}, \citet{white:miss}) that minimizes the ``distance''
or \emph{discrepancy} between the data and the moment constraints implied by
the model. This approach has gained further support since the introduction
of Generalized Empirical Likelihoods (GEL) and Minimum Discrepancy
estimators (\citet{owen:el1}, \citet{qinlawless}, \citet{newey:gel}), all of
which provide more readily interpretable pseudo-true values (%
\citet{imbens:res}, \citet{kitamura:itgmm}, \citet{schennach:etel}).

GEL implicitly
attributes the mismatch in the moment conditions solely to a biased sampling of the population. While this
is a possible explanation, it is not the only reason a valid model would
fail overidentification tests, when taken to the data. Another natural
possibility is the presence of errors in the variables (\citet{aguiar:revpref}, %
\citet{doraszelski:megmm}, \citet{schennach:HB}). In this work, we develop
an alternative to GMM that ensures, by construction, that overidentifying
restrictions are satisfied by allowing for possible errors in the
variables instead of sampling bias. We employ the generic term {\em error} to include, not only measurement error, but anything that could cause the recorded data to differ from the value they should have if the model were fully correct, i.e., this could include some model errors.
More generally, we allow distortions in the data generating process whose magnitude is quantified by a Wasserstein-type metric (\citet{villani:new}),
in the spirit of distributionally robust methods (\citet{connault:sens}, \citet{blanchet:distrob}).
In analogy with GEL, which does
not require the form of the sampling bias to be explicitly specified, the
error process does not need to be explicitly specified in our
approach, but is instead inferred from the requirement of satisfying the
overidentifying constraints imposed by the GMM model.
Of course, the accuracy of the
resulting estimated parameters will typically improve with the degree of overidentification.%
\footnote{In the case where the statistical properties of the errors are in fact
known a priori, other methods may be more appropriate (e.g. \citet{schennach:nlme}, \citet{schennach:elvis}),
\citet{schennach:annrev}, \citet{schennach:HB} and references therein.}

A fruitful way to accomplish this is to employ concepts from the general
area of optimal transport problems (e.g., \citet{galichon:optran}, %
\citet{villani:new}, \citet{carlier:vectq}, \citet{galichon:vectquan}, \citet{cherno:breniercons}, %
\citet{schennach:nlfact}). The idea is to find the parameter value that
minimizes cost of \textquotedblleft transporting\textquotedblright\ the observed
distribution of the data $\mu_x$ onto another distribution $\mu_z$ that satisfies the moment
conditions exactly.
Formally, the true iid data $z_{i}$ is assumed to satisfy $\mathbb{E}\left[ g\left(
z_{i},\theta \right) \right] =0$, where $\mathbb{E}$ is the expectation
operator, for a parameter value $\theta$ in some set $\Theta$ and some given $d_{g}$-dimensional vector $%
g\left( z_{i},\theta \right) $ of moment functions. However, we instead
observe an error-contaminated counterpart $x_{i}$ of the true vector $z_{i}$ (both taking value in $\mathcal{X} \subseteq \mathbb{R}^{d_{x}}$). We seek to exploit the model's over-identification to gain information regarding the error in $x_{i}$.
The Euclidean norm $\Vert (z-x) \Vert$ is chosen here for computational convenience, although one could imagine a whole class of related estimators obtained with different choices of metric. Our focus on Euclidean norms parallels the
choice made in common estimators (e.g. least squares regressions,
classical minimum distance and even GMM). Considering a {\it weighted} Euclidean norm can be useful to indicate the relative expected error magnitudes along different dimensions of $x$.

Given a probability measure $\mu_x$ for the random variable $x$, this setup suggests solving the following population optimization problem, for a given $\theta$: 
\begin{equation}
\min_{\mu _{zx}}\mathbb{E}_{\mu _{zx}}\left[ \left\Vert z-x\right\Vert ^{2} \label{eqobjgen}%
\right]
\end{equation}
subject to $\mu _{zx}$, supported on $\mathcal{X}\times \mathcal{X}$, having marginal $\mu _{x}$ and $\mathbb{E}_{\mu _{zx}} \left[ g\left( z,\theta \right) \right] =0$, where $\mathbb{E}_{\mu}$ denotes an expectation under the measure $\mu$. (This problem is guaranteed to have a solution if there exists at least one measure $\mu_z^*$ such that $\mathbb{E}_{\mu_z^*} \left[ g\left( z,\theta \right) \right] =0$.)
This setup covers the most general case, including both discrete and continuous variables, and can be handled using linear programming techniques
(e.g., \citet{santambrogio:optran}, Section 6.4.1), after observing that the moment constraint is easy to incorporate since it is linear in the probability measure. However, we shall focus on the purely continuous case in the remainder of this paper, because it enables us to express the main ideas more transparently. The fully continuous case indeed admits a convenient treatment, under the following regularity condition:
\begin{condition}
The marginals $\mu_z$ (arising from the solution $\mu_{zx}$ at each $\theta\in\Theta$) and $\mu_x$ have finite variance and $\mu_x$ is absolutely continuous with respect to the Lebesgue measure.
\end{condition}
Under this condition, by Theorem 1.22 in \citet{santambrogio:optran}, there exists a unique $\mu_{zx}$ implied by a deterministic transport map $z=q(x)$ that solves the constrained optimization problem (\ref{eqobjgen}) and yielding a transport cost $\mathbb{E}_{\mu_x} \left[ \left\Vert q(x)-x\right\Vert ^{2}\right]$. 
Since determining the function $q$ amounts to finding which value $z$ each point $x$ should be mapped to, the sample version of this problem can be stated as
\begin{equation}
\min_{\left\{ z_{i}\right\} }\frac{1}{2}\hat{\mathbb{E}}\left[ \left\Vert
z-x\right\Vert ^{2}\right]  \label{eqobjf}
\end{equation}%
subject to:%
\begin{equation}
\hat{\mathbb{E}}\left[ g\left( z,\theta \right) \right] =0,  \label{eqcons}
\end{equation}%
where $\hat{\mathbb{E}}$ denotes sample averages (i.e. $\hat{%
\mathbb{E}}[a\left( x\right) ]\equiv \frac{1}{n}\sum_{i=1}^{n}a\left(
x_{i}\right) $, where $n$ is sample size and $a(\cdot)$ a given function).
This optimization problem is then nested into an optimization over $\theta $,
which delivers the estimated parameter value $\hat{\theta}$. We call this
estimator an \emph{Optimally-Transported}\ GMM (OTGMM) estimator.

Our approach is conceptually similar to GEL, in that it
minimizes some concept of distributional distance under moment constraints.
Yet, the notion of distance used differs significantly. As shown
in Figure \ref{figotgmm}, the distance here is measured along the
\textquotedblleft observation values\textquotedblright\ axis rather than the
\textquotedblleft observation weights\textquotedblright\ axis (as it would
be in GEL). This feature arguably makes the method a hybrid between GEL and
optimal transport, since GEL's goal of satisfying all the moment
conditions is achieved through optimal transport instead of optimal reweighting.
(The distinction from GEL applies to the discrete case as well,
since the OTGMM objective function depends on both the amount and location of probability mass transfers,
while the GEL objective function is only sensitive on the amount, but not on the specific locations, of the probability mass transfers.)
Most of our regularity conditions will parallel those of optimal GMM and GEL, but
some will not (they are tied to the optimal transport nature of the problem and involve assumptions regarding higher order derivatives).
In analogy with the behavior of GEL estimators, the OTGMM estimator 
will be shown to be root-$n$ consistent and asymptotically normal, despite involving an optimization problem having 
an infinite-dimensional nuisance parameter (the $z_i$).
In general, however, OTGMM's asymptotic variance does not coincide with that of GEL or efficiently weighted GMM under correct specification.
Hence, OTGMM is most useful when errors in the variables or Wasserstein-type deviations in the data distribution constitute a primary concern.

\begin{figure}
\centerline{\includegraphics[width=3.5in]{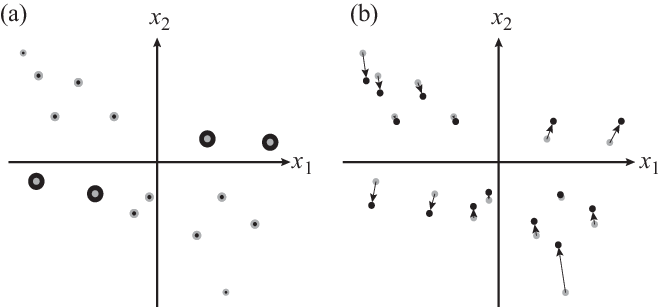}}
\caption{Comparison between the sample points adjustments for
(a) Generalized Empirical Likelihood (GEL), where observation weights (shown by point size) are adjusted, 
and (b) Optimal Transport GMM, where point positions are adjusted.
Simple case of overidentified (parameter-free) model imposing no correlation shown, with original sample in gray
and adjusted sample in black.}
\label{figotgmm}
\end{figure}

The remainder of the paper is organized as follows. We first formally define
and solve the optimization problem defining our estimator, before
considering the limit of small errors (in the spirit of %
\citet{chesher:1991})\ to gain some intuition. We then derive the resulting
estimator's asymptotic properties for the general, large error, case.
We show asymptotic normality and root $n$ consistency in both cases.
We then discuss related approaches and some extensions.
The method's practical usefulness is illustrated by revisiting 
an influential study of the relationship between a city's exports and the extent of its transportation infrastructure (\citet{duranton2014roads}).
Our results corroborate that study under weaker assumptions and provide insight into the error structure of the variables. 

\section{The estimator}

\subsection{Definition}

The Lagrangian associated with the constrained optimization problem defined
in Equations (\ref{eqobjf}) and (\ref{eqcons}) is%
\begin{equation*}
\frac{1}{2}\hat{\mathbb{E}}\left[ \left\Vert z-x\right\Vert ^{2}\right]
-\lambda ^{\prime }\hat{\mathbb{E}}\left[ g\left( z,\theta \right) \right],
\end{equation*}%
where $\hat{\mathbb{E}}[\ldots ]$ denotes a sample average and where $%
\lambda $ is a Lagrange multiplier. The dual problem's first-order conditions 
with respect to $\theta $, $\lambda $ and $z_{j}$, respectively, are
then%
\begin{eqnarray}
\hat{\mathbb{E}}\left[ \partial_{\theta}g^{\prime }(z,\theta )\right] \lambda
&=&0  \label{eqfoctheta} \\
\hat{\mathbb{E}}\left[ g(z,\theta )\right] &=&0  \label{eqfoclambda} \\
\left( z_{j}-x_{j}\right) -\partial_{z}g^{\prime }(z_{j},\theta )\lambda
&=&0\text{ for }j=1,\ldots ,n  \label{eqfoczj}
\end{eqnarray}%
where we let $\partial_{v}$ denote a partial derivative with respect to argument
$v$. We shall use $\partial_{v^{\prime }}$ to
denote a matrix of partial derivatives with respect to a transposed variable
(e.g., $\partial_{\theta^{\prime }}g(z,\theta )\equiv \partial g(z,\theta
)/\partial \theta ^{\prime }$). This formulation of the problem assumes
differentiability of $g\left( z,\theta \right) $ to a sufficiently high
order, as shall be formalized in our asymptotic analysis.

\subsection{Implementation}

The nonlinear system (\ref{eqfoctheta})-(\ref{eqfoczj}) of equations can be
solved numerically. To this effect, we propose an iterative procedure to
determine the $z_{j}$, $\lambda $ for a given $\theta $. This yields an\
objective function $\hat{Q}(\theta )$ that can be minimized to estimate $%
\theta $. Let $z_{j}^{t}$ and $\lambda ^{t}$ denote the approximations
obtained after $t$ steps. As shown in Supplement \ref{appiter}, given
tolerances $\epsilon ,\epsilon ^{\prime }$ and a given $\theta $, the
objective function $\hat{Q}(\theta )$ can be determined as follows:

\begin{algorithm}
\label{algoiter}

\begin{enumerate}
\item Start the iterations with $z_{j}^{0}=x_{j}$ and $t=0$.

\item Let $
\lambda ^{t+1}=\left( \hat{\mathbb{E}}\left[ H\left( z^{t},\theta \right)
H^{\prime }\left( z^{t},\theta \right) \right] \right) ^{-1}\left( -\hat{%
\mathbb{E}}\left[ g\left( z^{t},\theta \right) \right] +\hat{\mathbb{E}}%
\left[ H\left( z^{t},\theta \right) \left( z^{t}-x\right) \right] \right)$ and
$z_{j}^{t+1}=x_{j}+H^{\prime }\left( z_{j}^{t},\theta \right) \lambda ^{t+1}$,
where $H\left( z,\theta \right) =\partial_{z^{\prime }}g\left( z,\theta
\right) =\left( \partial_{z}g^{\prime }\left( z,\theta \right) \right)
^{\prime }$.

\item Increment $t$ by $1$; repeat from step 2 until $\left\Vert
z_{j}^{t+1}-z_{j}^{t}\right\Vert \leq \epsilon $ and $\left\Vert \lambda
^{t+1}-\lambda ^{t}\right\Vert \leq \epsilon ^{\prime }$.

\item The objective function is then: 
$\hat{Q}(\theta )=\frac{1}{2}(\lambda ^{t})^{\prime }\hat{\mathbb{E}}\left[
H\left( z^{t},\theta \right) H^{\prime }\left( z^{t},\theta \right) \right]
\lambda ^{t}\text{.}$

\end{enumerate}
\end{algorithm}

This algorithm is obtained by substituting $z_{j}=x_{j}+H^{\prime }\left(
z_{j},\theta \right) \lambda $ obtained from Equation (\ref{eqfoczj}) into
Equation (\ref{eqfoclambda}) and expanding the resulting expression to
linear order in $\lambda $. This linearized expression provides an improved
approximation $\lambda ^{t}$ to the Lagrange multiplier which can, in turn,
yield an improved approximation $z_{j}^{t}$. The process is then iterated to
convergence. The expression for $\hat{Q}(\theta )$ is obtained by
re-expressing $\hat{\mathbb{E}}\left[ \left\Vert z-x\right\Vert ^{2}\right] $
using Equation (\ref{eqfoczj}). Formal sufficient conditions for the
convergence of this iterative procedure can be found in Supplement \ref%
{appiterconv}. In cases where this simple approach fails to converge, one can employ
robust schemes based on a combination of discretization and linear programming
(see Section 6.4.1 in \citet{santambrogio:optran}).

To gain some intuition regarding the estimator, it is useful to consider the
limit of small errors when solving Equations (\ref{eqobjf})-(\ref%
{eqcons}), in the spirit of \citet{chesher:1991}).
This limit corresponds to assuming that higher-order powers of $\Vert z_i-x_i \Vert$
are negligible relative to $\Vert z_i-x_i \Vert$ itself.
In this limit, the estimator admits a closed form with an
intuitive interpretation, as shown by the following result, shown in
Appendix \ref{applin}.

\begin{proposition}
\label{propsml}To the first order in $z_{i}-x_{i}$ ($i=1,\ldots ,n$) the
estimator is equivalent to minimizing a GMM-like objective function with
respect to $\theta $ with a specific choice of weighting matrix:%
\begin{equation}
\hat{\theta}=\arg \min_{\theta }\hat{\mathbb{E}}\left[ g^{\prime }\left(
x,\theta \right) \right] \left( \hat{\mathbb{E}}\left[ H\left( x,\theta
\right) H^{\prime }\left( x,\theta \right) \right] \right) ^{-1}\hat{\mathbb{%
E}}\left[ g\left( x,\theta \right) \right] .  \label{eqsmlerr}
\end{equation}
\end{proposition}

From this expression, it is clear that the estimator downweights the
moments that are the most sensitive to errors in $x$, as measured by $%
H\left( x,\theta \right) \equiv \partial_{z^{\prime }}g\left( x,\theta
\right) $. This accomplishes the desired goal of minimizing the effect of
the errors when the properties of the error process are unknown.

Although this weighting matrix appears suboptimal (relative to a correctly specified optimally weighted
GMM estimator), one should realize that the notion of optimality depends on what class of data generating processes the ``true model'' encompasses.
Optimal GMM and GEL can be seen as Maximum Likelihood estimators (\citet{chamberlain:gmmeff}, \citet{newey:gel}) under moment conditions expressed in terms of the observed $x$.
In contrast, OTGMM can be interpreted as a Maximum Likelihood estimator for homoskedastic and normally distributed errors ($x-z$) under moment constraints on the unobserved $z$.
There is therefore a clear efficiency-robustness trade-off: OTGMM is less efficient if the observed $x$ satisfy the moment constraints, but allows for additional error terms that maintain the model's correct specification even if the observed $x$ do not satisfy the overidentified moment constraints, a more general setting where optimally weighted GMM or GEL offers no efficiency guarantees.

\section{Asymptotics}

In this section, we show that, despite the estimator's roots in the theory
of optimal transport, its large sample behavior remains amenable to standard
asymptotic tools since our focus is on an estimator of the parameter $\theta $
rather than on an estimator of a distribution. We first consider the case
of small errors, a limiting case that may be especially important in the
relatively common case of applications where overidentifying restrictions
tests are near the rejection region boundary. This limit also parallels the
approach taken in the GEL literature, where asymptotic properties are often
derived in the case where the overidentifying restrictions hold (e.g., %
\citet{newey:gel}).

\subsection{Small errors limit}

Our small error results enable us to illustrate that there is little risk
in using our estimator instead of efficient GMM when one is concerned about
overidentification test failure. If the data were to, in fact,
satisfy the moment conditions, using our approach does not sacrifice
consistency, root $n$ convergence or asymptotic normality. The only possible
drawback would be a suboptimal weighting of overidentifying moment
conditions, potentially leading to an increase in variance if the model happened to be correctly specified.
Conversely, if the model is misspecified, e.g., because the data is error-contaminated, the optimal weighting of efficient GMM is no longer the optimal weighting (since random deviations due to sampling variability are not the main reason for the failure to simultaneously satisfy all moment conditions).
For instance, if the errors are such that there is an unknown bias in the moment conditions that decays to zero asymptotically but at a rate possibly slower than $n^{-1/2}$, then the model is still correctly specified asymptotically but the bias dominates the random sampling error. Then, the optimal weighting should seek to minimize the effect of error-induced bias, which our approach seeks to accomplish by weighting based on the effect of errors in the variables on the moment conditions.
Hence, in that sense, the method provides a complementary alternative to standard GMM estimation offering a different trade-off between efficiency and robustness to misspecification.

Our consistency result requires a number of fairly standard primitive
assumptions.
\begin{condition}
\label{assiid}The random variables $x_{i}$ are iid and take
value in $\mathcal{X}\subset \mathbb{R}^{d_{x}}$.
\end{condition}

\begin{condition}
\label{ass1sml}$\mathbb{E}[g(x_{i};\theta _{0})]=0$, and $\mathbb{E}[g(x_{i};\theta )]\neq 0$
for other $\theta \in \Theta $, a compact set.
\end{condition}

In other words, Assumption \ref{ass1sml} indicates that we consider here the
case where GMM would be consistent, in analogy with the setup traditionally considered in the GEL literature (e.g. \citet{newey:gel}).

\begin{condition}
\label{ass2sml}$\mathbb{V}[g(x_{i};\theta _{0})]<\infty $, where $\mathbb{V}$
denotes the variance operator.
\end{condition}

\begin{condition}
\label{ass5sml}$g(x;\cdot )$ is almost surely continuous and $\Vert
g(x;\theta )\Vert\leq h(x)$ for any $\theta \in \Theta $ and for some function $%
h$ satisfying $\mathbb{E}\left[ h\left( x_{i}\right) \right] <\infty $.
\end{condition}

While Assumptions \ref{assiid}, \ref{ass1sml}, \ref{ass2sml} and \ref{ass5sml}
directly parallel those needed to establish the asymptotic properties of a
standard GMM estimator (e.g. Theorems 2.6 and 3.2 in \citet{Newey:HB}),
our estimator requires a few more low-level regularity conditions.
Given that our estimator, in the small error limit (Equation (\ref%
{eqsmlerr})), involves a sample average involving derivative $H\left(
x,\theta \right) \equiv \partial_{z^{\prime }}g\left( x,\theta \right) $,
we need to place some constraints on the behavior of that quantity as well.
Below, we let $\left\Vert a\right\Vert =\left(\sum_{i,j}a_{i,j}^2\right) ^{1/2}$ for a matrix $a$.

\begin{condition}
\label{ass3sml}$g$ is differentiable in its first argument and the
derivative satisfies $\mathbb{E}[\Vert \partial_{z^{\prime }}g(x_{i};\theta
_{0})\Vert ^{2}]\allowbreak <\infty $. Moreover, $\Vert \partial_{z^{\prime
}}g(x_{i};\theta )\Vert \neq 0$ almost surely for all $\theta \in \Theta $.
\end{condition}

\begin{condition}
\label{ass4sml}$\partial_{z^{\prime }}g(x;\theta _{0})$ is H\"{o}lder
continuous in $x$.
\end{condition}


\begin{condition}
\label{ass7sml}$\mathbb{E}[\partial_{z^{\prime }}g(x_{i};\theta
_{0})\partial_{z}g^{\prime }(x_{i};\theta _{0})]$ exists and is of full rank.
\end{condition}

These assumptions ensure that the minimization problem defined by (\ref%
{eqobjf}) and (\ref{eqcons}) is well-behaved, i.e., small changes in the
values of $x_{i}$ do not lead to jumps in the
solution $z_{i}$ to the optimization problem (aside from zero-probability events). It is likely that these
assumptions can be relaxed using empirical processes techniques. However,
here we favor simply imposing more smoothness (compared to the standard GMM
assumptions), because this leads to more transparent assumptions. They can
all be stated in terms of the basic function $g(x;\theta )$ that defines the
moment condition model, making them fairly primitive. We can then state our
first consistency result.

\begin{theorem}
\label{thconsissml}Under assumptions \ref{assiid}-\ref{ass7sml}, the OTGMM
estimator is consistent for $\theta _{0}$ and $\lambda =O_{p}(n^{-1/2})$.
\end{theorem}

As a by-product, this theorem also secures a convergence rate on the
Lagrange multiplier $\lambda $ which proves useful for establishing our
distributional results. The conditions needed to show asymptotic normality
also closely mimic those of standard GMM estimators (e.g. Theorem 3.2 in 
\citet{Newey:HB}):

\begin{condition}
\label{ass8sml}$\theta _{0}\in \Theta ^{\circ }$, the interior of $\Theta $.
\end{condition}

\begin{condition}
\label{ass9sml}$\mathbb{E}[\sup_{\theta \in \eta }\Vert \partial_{\theta^{\prime
}}g(x_{i};\theta )\Vert ]<\infty $ where $\eta \subset \Theta $ is a
neighborhood of $\theta _{0}$.
\end{condition}

\begin{condition}
\label{ass10sml} $\left( \mathbb{E}[\partial_{\theta'}g(x_{i};\theta _{0})^{\prime
}]\left( \mathbb{E}[\partial_{z'}g(x_{i};\theta _{0})\partial
_{z'}g(x_{i},\theta _{0})^{\prime }]\right) ^{-1}\mathbb{E}[\partial
_{\theta'}g(x_{i};\theta _{0})]\right) $ is invertible.
\end{condition}

We can then provide an explicit expression of the estimator's asymptotic variance.

\begin{theorem}
\label{thnormsml}Under Assumptions \ref{assiid}-\ref{ass10sml}, the OTGMM
estimator is asymptotically normal with $\sqrt{n}(\hat{\theta}%
_{OTGMM}-\theta _{0})\rightarrow ^{d}\mathcal{N}(0;V)$, where 
\begin{eqnarray*}
V &=&\left( \mathbb{E}[G_i^{\prime }]\left( \mathbb{E}[H_i H_i^{\prime
}]\right) ^{-1}\mathbb{E}[G_i]\right) ^{-1}\times 
(\mathbb{E}[G_i^{\prime }]\left( \mathbb{E}[H_iH_i^{\prime }]\right)
^{-1}\mathbb{E}[g_ig_i^{\prime }]\left( \mathbb{E}[H_i H_i ^{\prime }]\right)
^{-1}\mathbb{E}[G_i])\times \\
&&\left( \mathbb{E}[G_i^{\prime }]\left( \mathbb{E}[H_i H_i^{\prime
}]\right) ^{-1}\mathbb{E}[G_i]\right) ^{-1},
\end{eqnarray*}
where\ $H_i \equiv \partial_{z^{\prime }}g(x_{i};\theta _{0})$,
       $G_i \equiv \partial_{\theta^{\prime }}g(x_{i};\theta _{0})$ and
       $g_i \equiv g(x_{i};\theta _{0})$.
\end{theorem}

This simple normal limiting distribution with root $n$ convergent behavior is somewhat unexpected from an estimator that involves a high-dimensional optimization over $O(n)$ latent variables. As in GEL, this is made possible thanks to the existence of an equivalent low-dimensional dual optimization problem, which is, in turn, equivalent to a simple GMM estimator, albeit with a nonstandard weighting matrix. Thus, the variance has the expected \textquotedblleft sandwich\textquotedblright\ form, since the reciprocal weights $\mathbb{E}[H_i H_i^{\prime }]$
differs from the moment variance $\mathbb{E}[g_i g_i^{\prime }]$. For comparison, an optimally weighted GMM estimator would have an asymptotic variance of $(\mathbb{E}[G_i]'(\mathbb{E}[g_ig_i'])^{-1}\mathbb{E}[G_i])^{-1}$.

It is difficult to formulate a simple expression for the difference $V_{OTGMM}-V_{GMM}$ between the asymptotic variance of OTGMM and that of optimal GMM, but a simple example suffices to illustrate that this difference could be any non-negative-definite matrix.

Let $x_{i}$ take value in $\mathbb{R}^{d_{x}}$ ($d_{x}\geq 2$) with $\mathbb{E}\left[ x_{i\ell }\right] =0$ and $\mathbb{E}\left[ x_{i\ell }x_{i\ell ^{\prime }}\right] \equiv \omega _{\ell } $ when $\ell=\ell'$ and zero otherwise, for $\ell ,\ell ^{\prime }=1,\ldots ,d_{x}$. For $\theta \in \mathbb{R}$, let the moment functions be given by $g_{\ell }\left( x_{i},\theta \right) =x_{i\ell }-\theta $. We then have
\[
V_{OTGMM}=\frac{1}{d_{x}^{2}}\sum_{\ell =1}^{d_{x}}\omega _{\ell }\text{ and }V_{GMM}=\left( \sum_{\ell =1}^{d_{x}}\omega _{\ell }^{-1}\right) ^{-1}.
\]
Then, by the harmonic-arithmetic mean inequality, $V_{OTGMM}\geq V_{GMM}$ with equality when $\omega _{\ell }=c$ for all $\ell $. If $\omega _{1}\longrightarrow \infty $ leaving all the other $\omega _{\ell }$ finite, $V_{OTGMM}\longrightarrow \infty $ while $V_{GMM}$ remains finite. (General non-negative-definite differences can be obtained by stacking such moment conditions for different parameters $\theta_k$ ($k=1,\ldots,d_\theta$) and possibly linearly transforming the resulting parameter vector $\theta$.) Generally, we expect the difference to be large when some moments have a disproportionally large variance while not being disproportionally sensitive to changes in the underlying variables. 

\subsection{Asymptotics under large errors}

In some applications, there may be considerable misspecification or its magnitude may be a priori unknown. It thus proves useful to relax the assumption of small errors in deriving the estimator's asymptotic properties. To handle this more general setup, we employ the following equivalence, demonstrated in Appendix \ref{appproof}.

\begin{theorem}
\label{thjidgmm}If $g\left( z,\theta \right) $ is differentiable in its
arguments, the OTGMM estimator is equivalent to a just-identified GMM
estimator expressed in terms of the modified moment function 
\begin{equation}
\tilde{g}\left( x,\theta ,\lambda \right) =\left[ 
\begin{array}{c}
\partial_{\theta}g^{\prime }\left( q\left( x,\theta ,\lambda \right) ,\theta
\right) ~\lambda \\ 
g\left( q\left( x,\theta ,\lambda \right) ,\theta \right)%
\end{array}%
\right]  \label{qauggmm}
\end{equation}%
that is a function of the observed data $x$ and the augmented parameter
vector $\tilde{\theta}\equiv \left( \theta ^{\prime },\lambda ^{\prime
}\right) ^{\prime }$ and where 
\begin{equation}
q\left( x,\theta ,\lambda \right) \equiv
\arg \min_{z:z-\partial_{z}g^{\prime }\left( z,\theta \right) \lambda =x}\left\Vert z-x\right\Vert
^{2}.  \label{eqinvfoclam}
\end{equation}
\end{theorem}

Note that $q\left( x,\theta ,\lambda \right) $ is essentially the inverse of
the mapping $z-\partial_{z}g^{\prime }\left( z,\theta \right) \lambda =x$
(from Equation (\ref{eqfoczj})), augmented with a rule to select the
appropriate branch in case the inverse is multivalued.

The equivalence result of Theorem \ref{thjidgmm} implies that many of the 
asymptotic technical tools used in GMM-type estimators can be adapted to our
setup, with the distinction that the function $q\left( x,\theta ,\lambda
\right) $ is defined only implicitly. Hence, many of our efforts below seek
to recast necessary conditions on $q\left( x,\theta ,\lambda \right) $ in
terms of more primitive conditions on the moment function $g\left( z,\theta
\right) $ whenever possible or in terms of regularity conditions drawn from optimal transport theory.

We start with a standard GMM-like identification condition:

\begin{condition}
\label{assidgmm}For some compact sets $\Theta $ and $\Lambda $, there exists
a unique $\left( \theta _{0},\lambda _{0}\right) \in \Theta \times \Lambda $
solving $\mathbb{E}\left[ \tilde{g}\left( x,\theta ,\lambda \right) \right]
=0$ for $\tilde{g}\left( x,\theta ,\lambda \right) $ defined in Theorem \ref%
{thjidgmm}.
\end{condition}

This condition is implied by a natural uniqueness and regularity condition on the solution to the primal optimal transportation problem (Equations (\ref{eqobjgen})):

\begin{condition}\label{asscaff}
Let $\mu_{zx;\theta}$ denote the solution to Problem (\ref{eqobjgen}) for a given $\theta\in\Theta$.
(i) $\mathbb{E}_{\mu_{zx;\theta}}[\Vert z-x \Vert^2]$ is uniquely minimized at $\theta=\theta_0$
(ii) The corresponding marginals $\mu_{z;\theta}$ and $\mu _{x}$ are absolutely continuous with respect to the Lebesgue measure with a density that is
finite, nonvanishing and H\"{o}lder continuous on their convex support.
\end{condition}

Indeed, by Theorem [C3], part b) and d), in \citet{caffarelli:regconv}, Assumption \ref{asscaff}(ii) implies that, at each $\theta$, there exists a unique invertible transport map $z=q\left( x\right) $ from $\mu _{x}$ to $\mu _{z;\theta}$ and both $q$ and its inverse are equal to the gradient of a twice differentiable strictly convex function. The fact that $q^{-1}$ is the gradient of a twice differentiable strictly convex function ensures that the first-order condition in Assumption \ref{assidgmm} has a unique solution (making a rule to handle multivalued inverses unnecessary).


Next, we consider standard continuity and dominance conditions that are used
to establish uniform convergence of the GMM objective function.
These assumptions constitute a superset of those needed for standard GMM because the
modified moment conditions include the additional parameter $\lambda$ and higher-order derivatives of the original moment conditions.
In a high-level form, these conditions read:

\begin{condition}
\label{asssimplegtilda}(i) $\tilde{g}\left( x,\theta ,\lambda \right) $ is
continuous in $\theta $ and $\lambda $ for $\left( \theta ,\lambda \right)
\in \Theta \times \Lambda $ with probability one and (ii) $\mathbb{E}\left[
\sup_{\left( \theta ,\lambda \right) \in \Theta \times \Lambda }\left\Vert 
\tilde{g}\left( x,\theta ,\lambda \right) \right\Vert \right] <\infty $.
\end{condition}

Alternatively, Assumption \ref{asssimplegtilda} can be replaced by more
primitive conditions on $g\left( z,\theta \right) $ instead, as given below
in Assumptions \ref{assgcontdiff}, \ref{asseigval} and \ref{assgdom}.

\begin{condition}
\label{assgcontdiff}(i) $g\left( z,\theta \right) $ and
$\partial_{z^{\prime }}g\left( z,\theta \right) $ are differentiable in $\theta $ and
(ii) $\partial_{\theta^{\prime }}g\left( z,\theta \right) $ is continuous in
both arguments.
\end{condition}

This assumption parallels continuity assumptions typically made for GMM, but
higher order derivatives of $g\left( z,\theta \right) $ are needed, because
they enter the moment condition either directly or indirectly via the
function $q\left( x,\theta ,\lambda \right) $.
The next condition ensures that the function $q\left( x,\theta ,\lambda
\right) $ is well behaved.

\begin{condition}
\label{asseigval}$\bar{\nu}\bar{\lambda}<1$ where $\bar{\lambda}%
=\max_{\lambda \in \Lambda }\left\Vert \lambda \right\Vert $ and $\bar{\nu}%
=\sup_{\theta \in \Theta }\sup_{z\in \mathcal{X}}\max_{k\in \left\{ 1,\ldots
,d_{g}\right\} }\allowbreak \max \eigval\left( \partial_{zz^{\prime
}}g_{k}\left( z,\theta \right) \right) $, in which $\partial_{zz^{\prime
}}g_{k}\left( z,\theta \right) $ exists for $k=1,\ldots ,d_{g}$ and where $%
\eigval\left( M\right) $ for some matrix $M$ denotes the set of its
eigenvalues.
\end{condition}

Once again, this condition can be alternatively phrased in terms of optimal transport concepts.
The first order condition which implicitly defines $z=q\left( x,\theta
,\lambda \right)$ can be written in terms of the derivative of a \emph{potential function}
$\psi(z,\theta,\lambda)=z^{\prime }z/2-g^{\prime }\left(z,\theta \right) \lambda$: $\nabla_z \psi(z,\theta,\lambda)=x$.
With the help of Theorem [C3], part b) and d), in \citet{caffarelli:regconv}, Assumption \ref{asscaff}(ii) implies that, at each $\theta$, the above
potential $\psi(z,\theta,\lambda)$ has a positive-definite Hessian (with respect to $z$), which implies Assumption \ref{asseigval}.

In order to state our remaining regularity conditions, it is useful to
introduce a notion of (nonuniform) Lipschitz continuity, combined with dominance conditions.

\begin{definition}
\label{defholderLD}Let $\mathcal{L}$ be the set of functions $h\left(
z,\theta \right) $ such that (i) $\mathbb{E}\left[ \sup_{\theta \in \Theta
}\left\Vert h\left( x,\theta \right) \right\Vert \right]\linebreak <\infty $ and (ii)
there exists a function $\bar{h}\left( x,\theta \right) $ satisfying 
\begin{eqnarray}
\mathbb{E}\left[ \sup_{\theta \in \Theta }\bar{h}\left( x,\theta \right)
\left\Vert \partial_{z^{\prime }}g\left( x,\theta \right) \right\Vert %
\right] &<&\infty .  \label{eqdomhg} \\
\left\Vert h\left( z,\theta \right) -h\left( x,\theta \right) \right\Vert
&\leq &\bar{h}\left( x,\theta \right) \left\Vert z-x\right\Vert \label{eqnonulip}
\end{eqnarray}%
for all $x,z\in \mathcal{X}$ and $\theta\in\Theta$, and where $g\left( x,\theta \right) $ is as in the moment conditions.
\end{definition}

This Lipschitz continuity-type assumption has no parallel in conventional GMM.
It is made here because it ensures that
the behavior of the observed $x$ and the underlying unobserved $z$ will not
differ to such an extent that moments of unobserved variables would be
infinite, while the corresponding observed moments are finite. Clearly,
without such an assumption, observable moments would be essentially
uninformative. The idea underlying Definition \ref{defholderLD} is that we
want to define a property that is akin to Lipschitz continuity but that
allows for some heterogeneity (through the function $\bar{h}\left( x,\theta
\right) $ in Equation (\ref{eqnonulip})). This heterogeneity proves
particularly useful in the case where $\mathcal{X}$ is not compact (for
compact $\mathcal{X}$, one can take $\bar{h}\left( x,\theta \right) $ to be
constant in $x$ with little loss of generality). For a given function $%
h\left( x,\theta \right) $ that is finite for finite $x$, membership in $%
\mathcal{L}$ is easy to check by inspecting the tail behavior (in $x$) of
the given function $h\left( x,\theta \right) $. Polynomial tails will
suggest a polynomial form for $\bar{h}\left( x,\theta \right) $, for
instance. Equation (\ref{eqdomhg}) strengthens the dominance condition \ref%
{defholderLD}(i) to ensure that functions $h\left( x,\theta \right) $ in $%
\mathcal{L}$ also satisfy a dominance condition when interacted with other
quantities entering the optimization problem,
i.e. $\partial_{z^{\prime}}g\left( x,\theta \right) $).

With this definition in hand, we can succinctly state a sufficient condition
for $\tilde{g}\left( x,\theta ,\lambda \right) $ to satisfy a dominance
condition:

\begin{condition}
$\,$\label{assgdom}$g\left( \cdot ,\cdot \right) $ and each element of $%
\partial_{\theta}g^{\prime }\left( \cdot ,\cdot \right) $ belong to $\mathcal{L}$%
.
\end{condition}

We are now ready to state our general consistency result.

\begin{theorem}
\label{thconsislarge}Under Assumptions \ref{assiid}, \ref{assidgmm} and
either Assumption \ref{asssimplegtilda} or Assumptions \ref{assgcontdiff}, 
\ref{asseigval}, \ref{assgdom}, the OTGMM estimator is consistent
($(\hat{\theta},\hat{\lambda})\overset{p}{\longrightarrow }(\theta _{0},\lambda _{0}) $).
\end{theorem}

We now turn to asymptotic normality. We first need a conventional
\textquotedblleft interior solution\textquotedblright\ assumption.

\begin{condition}
\label{assinterlarge}$\left( \theta _{0},\lambda _{0}\right) $ from
Assumption \ref{assidgmm} lies in the interior of $\Theta \times \Lambda $.
\end{condition}

Next, as in any GMM estimator, we need finite variance of the moment
functions and their differentiability:

\begin{condition}
\label{asshesjaclarge}(i) $\mathbb{V}\left[ \tilde{g}\left( x,\theta
_{0},\lambda _{0}\right) \right] \equiv \Omega $ exists and (ii) $\mathbb{E}%
\left[ \partial \tilde{g}\left( x,\theta ,\lambda \right) /\partial \left(
\theta ^{\prime },\lambda ^{\prime }\right) \right] \equiv \tilde{G}$ exists
and is nonsingular.
\end{condition}

Assumption \ref{asshesjaclarge}(ii) can be expressed in a more primitive
fashion using the explicit form for $\tilde{G}$ provided in Theorem \ref%
{thanlarge} below.

Next, we first state a high-level dominance condition that ensures uniform
convergence of the Jacobian term $\partial \tilde{g}\left( x,\theta ,\lambda
\right) /\partial \left( \theta ^{\prime },\lambda ^{\prime }\right) $.

\begin{condition}
\label{assdgdom}(i) $\tilde{g}\left( x,\theta ,\lambda \right) $ is
continuously differentiable in $\left( \theta ,\lambda \right) $;\linebreak (ii) $%
\mathbb{E[}\sup_{\left( \theta ,\lambda \right) \in \Theta \times \Lambda
}\allowbreak \left\Vert \partial \tilde{g}\left( x,\theta ,\lambda \right)
/\partial \left( \theta ^{\prime },\lambda ^{\prime }\right) \right\Vert
]<\infty $.
\end{condition}

This assumption is implied by the following, more primitive, condition:

\begin{condition}
\label{assdgdomprim}(i) $g\left( z,\theta \right) $ and
$\partial_{\theta}g\left( z,\theta \right) $ are continuously differentiable in $\theta $,
(ii) all elements of $\partial_{\theta}g_{k}\left( z,\theta \right) $ and $%
\partial_{\theta\theta^{\prime }}g_{k}\left( z,\theta \right) $ for $k=1,\ldots
,d_{g} $ belong to $\mathcal{L}$ and (iii) Assumptions \ref{assgcontdiff}(i)
and \ref{asseigval} hold.
\end{condition}

We can now state our general asymptotic normality and root-$n$ consistency result, shown in Appendix \ref{appproof}.

\begin{theorem}
\label{thanlarge}Let the assumptions of Theorem \ref{thconsislarge} hold as
well as Assumptions \ref{assinterlarge}, \ref{asshesjaclarge} and either
Assumption \ref{assdgdom} or \ref{assdgdomprim}. Then, 
\begin{equation*}
\sqrt{n}\left( \left[ 
\begin{array}{c}
\hat{\theta} \\ 
\hat{\lambda}%
\end{array}%
\right] -\left[ 
\begin{array}{c}
\theta _{0} \\ 
\lambda _{0}%
\end{array}%
\right] \right) \overset{d}{\longrightarrow }\mathcal{N}\left(
0,W^{-1}\right)
\end{equation*}%
where $W=\tilde{G}^{\prime }\Omega ^{-1}\tilde{G}$, $\Omega =\mathbb{E}\left[
\tilde{g}\tilde{g}^{\prime }\right],$
\begin{equation*}
\tilde{g} \equiv \tilde{g}(z,(\theta_0,\lambda_0)) = \left[ 
\begin{array}{c}
\partial_{\theta}g^{\prime }\left( z,\theta_0 \right) \lambda_0 \\ 
g\left( z,\theta_0 \right)%
\end{array}%
\right] \text{ and }\tilde{G} \equiv \mathbb{E}[\partial_{\theta} \tilde{g}'] = \left[ 
\begin{array}{cc}
\tilde{G}_{\theta \theta } & \tilde{G}_{\theta \lambda } \\ 
\tilde{G}_{\lambda \theta } & \tilde{G}_{\lambda \lambda }%
\end{array}%
\right]
\end{equation*}%
in which%
\begin{eqnarray*}
\tilde{G}_{\theta \theta } &\equiv &\mathbb{E}\left[ \partial_{\theta \theta^{\prime
}}\left( \lambda _{0}^{\prime }g\left( z,\theta _{0}\right) \right)
+\partial_{\theta z^{\prime }}\left( \lambda _{0}^{\prime }g\left( z,\theta
_{0}\right) \right) \partial_{\theta^{\prime }}q\left( x,\theta _{0},\lambda
_{0}\right) \right] \\
\tilde{G}_{\lambda \theta } &\equiv &\mathbb{E}\left[ \partial_{\theta^{\prime
}}g\left( z,\theta _{0}\right) +\partial_{z^{\prime }}g\left(
z,\theta _{0}\right) \partial_{\theta^{\prime }}q\left( x,\theta
_{0},\lambda _{0}\right) \right] \\
\tilde{G}_{\theta \lambda } &\equiv &\mathbb{E}\left[ \partial_{\theta}\left(
g^{\prime }\left( z,\theta _{0}\right) \right) +\partial_{\theta z^{\prime
}}\left( \lambda _{0}^{\prime }g\left( z,\theta _{0}\right) \right)
\partial_{\lambda^{\prime }}q\left( x,\theta _{0},\lambda _{0}\right) \right]
\\
\tilde{G}_{\lambda \lambda } &\equiv &\mathbb{E}\left[ \partial_{z^{\prime
}}g\left( z,\theta _{0}\right) \partial_{\lambda^{\prime }}q\left(
x,\theta _{0},\lambda _{0}\right) \right]
\end{eqnarray*}%
where $z$ solves $x=z-\partial_{z}g^{\prime }\left( z,\theta \right) \lambda $ for given $x,\theta
,\lambda $ and where%
\begin{eqnarray}
\partial_{\theta^{\prime }}q\left( x,\theta ,\lambda \right) &=&\left[ \left(
I-\partial_{zz^{\prime }}\left( \lambda ^{\prime }g\left( z,\theta \right)
\right) \right) ^{-1}\partial_{z\theta^{\prime }}\left( \lambda ^{\prime
}g\left( z,\theta \right) \right) \right] _{z=q\left( x,\theta ,\lambda
\right) } \\
\partial_{\lambda^{\prime }}q\left( x,\theta ,\lambda \right) &=&\left[ \left(
I-\partial_{zz^{\prime }}\left( \lambda ^{\prime }g\left( z,\theta \right)
\right) \right) ^{-1}\partial_{z}g^{\prime }\left( z,\theta \right) \right]
_{z=q\left( x,\theta ,\lambda \right) }.
\end{eqnarray}
In particular, for $\theta $, the partitioned inverse formula gives%
\begin{equation*}
\sqrt{n}\left( \hat{\theta}-\theta _{0}\right) \overset{d}{\longrightarrow }%
\mathcal{N}\left( 0,\left( W_{\theta \theta }-W_{\theta \lambda }W_{\lambda
\lambda }^{-1}W_{\lambda \theta }\right) ^{-1}\right)
\end{equation*}%
where $W$ is similarly partitioned as:%
\begin{equation*}
W=\left[ 
\begin{array}{cc}
W_{\theta \theta } & W_{\theta \lambda } \\ 
W_{\lambda \theta } & W_{\lambda \lambda }%
\end{array}%
\right]
\end{equation*}
\end{theorem}

The asymptotic variance stated in Theorem \ref{thanlarge} takes the familiar
form expected from a just-identified GMM estimator: $(\tilde{G}^{\prime
}\Omega ^{-1}\tilde{G})^{-1}$. The relatively lengthy expressions merely come
from explicitly computing the first derivative matrix $\tilde{G}$ in terms
of its constituents. This is accomplished by differentiating $\tilde{g}$
with respect to all parameters using the chain rule and calculating the
derivative of $q\left( x,\theta ,\lambda \right) $ using the implicit
function theorem.

We thus have now completely characterized the first-order asymptotic properties of our estimator in the most general settings of large (i.e., non-local) misspecification. This result thus allows researcher to directly replace their GMM estimator which may happen to fail overidentification tests by another, logically consistent and easy-to-interpret, estimator where the overidentification failure is naturally accounted for by errors in the variables. In addition, researchers can further document the presence of errors via Theorem \ref{thanlarge}, as it enables, as a by-product, a formal test of the absence of error. Under this null hypothesis, which can be stated as $\lambda=0$, we have
\begin{equation}
    n\hat{\lambda}' \left( W_{\lambda\lambda} - 
W_{\lambda\theta}W_{\theta\theta}^{-1}W_{\theta\lambda} \right)^{-1} \hat{\lambda}
\overset{d}{\longrightarrow }
\chi^2_{d_g}, \label{eqchi}
\end{equation}
where the above expression can be straightforwardly derived from the partitioned inverse formula applied to the $\lambda$ sub-block of the asymptotic variance.

\section{Discussion and extensions}

On a conceptual level, our use of a so-called Wasserstein metric to measure distance between distributions does provide some desirable theoretical properties. 
For instance, the Wasserstein metric metrizes convergence in distribution (see Theorem 6.9 in \citet{villani:new}) under some simple bounded moment assumptions. In contrast, the discrepancies which generate GEL estimators do not admit such an interpretation. In fact, most discrepancies are not metrics, as they lack symmetry. The Kullback-Leibler discrepancy, which is perhaps the best known among them, does not allow comparison between distributions that are not absolutely continuous with respect to one another, whereas the Wassertein metric does. (Of course, leveraging this advantage requires considering the most general transport problem of Equation (\ref{eqobjgen}).) Finally, it is arguably logical to penalize probability transfer over larger distances more than the same probability transfer over smaller distances, as the Wasserstein metric does, while none of GEL-related discrepancies do.


An interesting extension of our approach would be a hybrid method in
which (i) the possibility of general forms of errors is
accounted for with the current method by constructing the equivalent GMM
formulation of the model via Theorem \ref{thjidgmm} and (ii) additional
restrictions on the form of the errors are imposed via additional
moment conditions involving some elements of $z$ and $x$. This could prove a
useful middle ground when a priori information regarding the errors
is available for some, but not all, variables.

As shown in Supplement \ref{secconstr}, it is straightforward to extend our approach to allow error in some, but not all variables.
When our method is used while allowing for errors in only a few variables,
it may not be possible to simultaneously satisfy all moment conditions for any $\theta$. In such cases, it could make sense to consider a hybrid method
where both errors in the variables, handled via our approach, and re-weighting of the
sample, handled via GEL, are simultaneously allowed.

Finally, we should mention other approaches aimed at handling violations of overidentifying restrictions, which
include the use of set identification combined with relaxed moment constraints (\citet{poirier:salvaging}),
placing a Bayesian prior on the magnitude of the deviations from correct specification of the moments (\citet{conley:plausexo}),
distributionally robust approaches that allow for deviations from the data generating process up to a given bound (\citet{connault:sens}, \citet{blanchet:distrob}), sensitivity analysis (\citet{shapiro:sens}, \citet{bonhomme:sens}) and misspecified moment inequality models (e.g. \citet{kwon:ineqmiss}).

\section{Application}

We revisit the study of \citet{duranton2014roads}, who documented evidence that the quantity of goods a city exports is strongly related to the extent of interstate highways present in that city.
Due to simultaneity concerns, the authors adopt an instrumental variable approach to recover the causal effect of building highways. Although the authors mitigate instrument validity concerns with controls, instrument exogeneity or exclusion could remain a concern (as noted by \citet{poirier:salvaging}) and this problem would manifest itself by specification test failure.

\begin{table}[!ht]
\caption{Main results}
\label{main_results}%
\begin{equation*}
\begin{array}{|l|l|l|}
\hline
& \mbox{GMM} & \mbox{OTGMM} \\ \hline
\mbox{log highway km} & 0.39 & 0.40 \\ 
\mbox{se} & (0.12) & (0.11) \\ \hline
\mbox{log employment} & 0.47 & 1.24  \\ 
\mbox{se} & (0.32) & (0.31) \\ \hline
\mbox{market access (export)} & -0.63 & -0.66  \\ 
\mbox{se} & (0.11) & (0.10) \\ \hline
\mbox{log 1920 population} & -0.29 & -0.57  \\ 
\mbox{se} & (0.23) & (0.23) \\ \hline
\mbox{log 1950 population} & 0.65 & 1.14  \\ 
\mbox{se} & (0.37) & (0.37) \\ \hline
\mbox{log 2000 population} & -0.20 & -1.25  \\ 
\mbox{se} & (0.44) & (0.35) \\ \hline
\mbox{log \% manuf empl} & 0.64 & 0.58  \\ 
\mbox{se} & (0.12) & (0.12) \\ \hline
\mbox{Overidentification P-value} & 0.30 & \\ \hline
\end{array}%
\end{equation*}
{Main results from Table 5 in \citet{duranton2014roads}. Original GMM estimates and OTGMM estimates. Heteroskedasticity-robust standard errors (GMM) and small-error asymptotic standard error (OTGMM) in parentheses.}
\end{table}

We apply our method to further assess the robustness of the results to potential model misspecification. We seek to recover point estimates that remain interpretable under potential misspecification and account for misspecification by viewing the model's variables as potentially measured with error. 
Not only do we consider errors in the regressors, but we also think of potentially invalid instruments as simply mismeasured versions of an underlying valid instrument that is unfortunately not available. Alternatively, one can think of the underlying valid instrument as the counterfactual value of the instrument in a world where the mechanism causing this instrument to be invalid would be absent. This broader interpretation of what could constitute an ``error'' under our framework considerably expands the scope of models that are conceptually consistent with our approach.

\citet{duranton2014roads} consider three instruments: (log) kilometers of railroads in 1898, quantity of historical exploration routes, and planned (log) highway kilometers according to a 1947 construction map.
The validity of these instruments could be criticized, for instance, in a situation where some cities are in proximity to key natural resources, which could cause higher exports and, at the same time, more transportation routes (or plans to build them). If this mechanism is active both in the present and in the past, the causal effect of highways on exports would be overestimated.
In our framework, the true but unavailable instrument could represent a measure of past transportation routes in a counterfactual world where natural resources would be evenly distributed among cities. The actual available instruments represent an approximation to this ideal instrument, a situation which we represent as a potentially non-classical errors-in-variables model.

The model's moment conditions are written in terms of
$g(z_i,\theta)=w_i(y_i-r_i'\theta)$, where the vector of observables
$z_i=(y_i,r_i',w_i')'$ contains the dependent variable $y_i$, the vector of regressors $r_i$ and the vector of instruments $w_i$.
In \citet{duranton2014roads}, the dependent variable $y_i$ measures ``propensity to export'' and is constructed from an auxiliary panel data model, which regresses volume of exports between given cities on distance and trading partner characteristics, modeled as fixed effects. It is the value of these fixed effects that is used as $y_i$. Following \citet{poirier:salvaging}, we take this construction as given and focus on export volume measured by weight.
Explicitly accounting for errors in $y_i$ is superfluous since, in a regression setting, errors in the dependent variable are already allowed for.

In the analysis below, $r_i$ always includes the regressor of main interest: the (logarithm of) the number of kilometers of highway. It also contains a number of controls, which may differ in the different models considered. These controls include: log employment, log market access, log population in 1920, 1950, and 2000, and log manufacturing share in 2003, average distance to the nearest body of water, land gradient, dummy variables for census regions, log share of the fraction of adult population with a college degree or more, log average income per capita, log share of employment in wholesale trade, and log average daily traffic on the interstate highways in 2005. 
We allow for errors in all of these variables, except for the constant term and dummies.




We consider two of the specifications employed by \citet{duranton2014roads}: The specification with many covariates that most obviously passes the overidentifying test and the specification with few covariates that most clearly fails this test.
(The other specifications fall in between these extreme cases and are thus not reported here, for conciseness.)
The results from GMM, replicated from the original study, and from OTGMM (allowing for large errors) are reported in Tables \ref{main_results} and \ref{failed_overID}. 

\begin{table}[!ht]
\caption{Specification that fails the test for over-identifying restriction}
\label{failed_overID}%
\begin{equation*}
\begin{array}{|l|l|l|}
\hline
& \mbox{GMM} & \mbox{OTGMM} \\ \hline
\mbox{log highway km} & 0.57 & 0.65 \\ 
\mbox{se} & (0.16) & (0.19) \\ \hline
\mbox{log employment} & 0.52 & 0.49  \\ 
\mbox{se} & (0.11) & (0.09) \\ \hline
\mbox{market access (export)} & -0.45 & -0.46  \\ 
\mbox{se} & (0.13) & (0.14) \\ \hline
\mbox{Overidentification P-value} & 0.043 & \\ \hline
\end{array}%
\end{equation*}
{Specification from column 2, Table 5, in \citet{duranton2014roads}. Replicated estimates from the original paper and OTGMM estimates. Heteroskedasticity-robust standard errors (GMM) and small-error asymptotic standard error (OTGMM) in parentheses.}
\end{table}

The OTGMM estimates are similar to those of \citet{duranton2014roads}. Although some coefficients of the controls are larger in magnitude, the main coefficient --- the elasticity of export weight relative to kilometers of highway --- is almost unchanged and remains statistically significant.
As in the original study, the 95\% confidence intervals of the main elasticity of interest obtained with different sets of controls overlap significantly.

It is also instructive to look at the correction of the underlying variables implied by OTGMM. Table \ref{MEs}, columns 1 and 2, documents this by looking at the standard deviation of different elements of the correction $z_{i} - x_{i}$. 
The corresponding quantities for the models of Tables \ref{main_results} and \ref{failed_overID} are reported in columns 1 and 2, respectively.
To get an idea of the scale, the last column reports the standard deviations of the observed variables.
As expected, the errors in column 1 are exceedingly small, reflecting the fact that the model passes the overidentification test.
In contrast, column 2 is particularly interesting for our purposes because it corresponds to a specification that fails the test of over-identifying restrictions. While this may, at first, cast doubt regarding the GMM estimates, OTGMM shows that errors of the order of only 10\% of the observed regressors' magnitude are sufficient to eliminate the mispecification. Such small error magnitudes are highly plausible empirically, thus supporting the plausibility of the OTGMM estimate. As the GMM and OTGMM estimates of the main elasticity of interest are very close in Table \ref{failed_overID}, this also corroborates the authors' GMM estimates. Overall, our approach strongly supports the conclusions of the original study and thus provides an effective robustness test.

The fact that the model with more controls leads to smaller error magnitudes is very consistent with our interpretation: to the extent that including more controls reduces the magnitude of the potentially endogenous residuals, one would expect that our estimator has to perform smaller alterations to the variables to arrive at valid instruments and/or non-endogeneous regressors. More quantitatively, suppose that the instrument $w$ and the error term $\epsilon$ can be decomposed into separate components, say $w=w_1+w_2$ and $\epsilon = \epsilon_1 + \epsilon_2$, where $w_1$ and $\epsilon_1$ are correlated, but the latter is well explained by additional controls. In a model with fewer controls, our method has to identify $w_1$ and/or $\epsilon_1$ as error components to obtain a correctly specified model. However, if including controls already accounts for the $\epsilon_1$ term, then our estimator no longer needs to correct the corresponding components and thus achieves orthogonality with smaller errors.

In contrast, using GEL to address misspecification would effectively assume that the misspecification originates from some form of selection bias. The fact that adding more control eliminates the misspecification indicates that the controls incorporate the variables that explain selection bias. Yet, these controls are not included in the model with few controls, precisely the model where the largest amount of sample reweighting would take place when using GEL. This paradox makes the use of GEL as a remedy difficult to rationalize in a setting where misspecification arises primarily from the unavailability of adequate control variables.



\begin{table}[!ht]
\caption{{Standard deviations of errors and corresponding regressors}}
\label{MEs}%
\begin{equation*}
\begin{array}{|l|l|l|l|}
\hline
\mbox{Variables} & \sqrt{\hat{\mathbb{V}}[z-x]} & \hfill \sqrt{\hat{\mathbb{V}}[z-x]} \hfill & \sqrt{\hat{\mathbb{V}}[x]}\\
& \mbox{(main)} &  \mbox{(failed over-ID)} & \\
\hline \hline
\mbox{log highway km} & 0.0034 & 0.0634 & 0.5884 \\ \hline
\mbox{log railroads km 1898} & 0.0048 & 0.0267 & 0.6031 \\ \hline
\mbox{exploration routes} & 0.0005 & 0.0418 & 0.8339 \\ \hline
\mbox{plan 1947} & 0.0075 & 0.0598 & 0.7049 \\ \hline
\mbox{log employment} & 0.0144 &  0.0473 & 0.8573 \\ \hline
\mbox{market access (export)} & 0.0056 & 0.0444 & 0.4864 \\ \hline
\mbox{log 1920 population} & 0.0052 & & 1.0417 \\ \hline
\mbox{log 1950 population} & 0.0097 & & 0.9253 \\ \hline
\mbox{log 2000 population} & 0.0163 & & 0.8083 \\ \hline
\mbox{log \% manuf empl} & 0.0050 & & 0.3707 \\ \hline
\end{array}%
\end{equation*}
\end{table}


In Supplement \ref{appallcon}, we perform a number of robustness tests along the lines suggested by \citet{poirier:salvaging}.
We further replicate specifications with many more controls and find that the GMM and OTGMM are still in close agreement.
In Supplement \ref{apprelaxex}, we also consider alternative specifications that allow for the possibility that some instrumental variables should be included as regressors. The performance of OTGMM in other simple models is also investigated via simulations in Supplement \ref{secsimul}.

\section{Conclusion}

We have proposed a novel optimal transport-based version of the Generalized
Method of Moment (GMM) that fulfills, by construction, the overidentifying
moment restrictions by introducing the smallest possible amount of error in the variables or, equivalently, by allowing for the smallest possible Wasserstein metric distortions in the data distribution. This approach
conceptually merges the Generalized Empirical Likelihood (GEL) and optimal
transport methodologies. It provides a theoretically motivated
interpretation to GMM results when standard overidentification tests reject
the null. GEL approaches handle model misspecification by re-weighting the
data, which would be appropriate when misspecification arise from improper
sampling of the population. In contrast, our optimal transport approach is
appropriate when the sources of misspecification are errors or, more generally, Wasserstein-type distortions in the data distribution, which
is arguably a common situation in applications.


\appendix

\section{Proofs}

\subsection{Linearized estimator}

\label{applin}

\begin{proof}[Proof of proposition \protect\ref{propsml}]
In the following, the approximation denoted by \textquotedblleft $%
\approx $\textquotedblright\ are exact to first order in $z_{j}-x_{j}$. In
that limit,\ $\partial_{z}g^{\prime }\left( z_{j},\theta \right) \approx
\partial_{z}g^{\prime }\left( x_{j},\theta \right) $. Therefore:%
\begin{eqnarray}
z_{j}-x_{j} &\approx &\partial_{z}g^{\prime }\left( x_{j},\theta \right)
\lambda  \notag \\
z_{j} &\approx &x_{j}+\partial_{z}g^{\prime }\left( x_{j},\theta \right)
\lambda  \label{eqgivez}
\end{eqnarray}%
Substituting into the constraint yields: 
$\sum_{j}g\left( x_{j}+\partial_{z}g^{\prime }\left( x_{j},\theta \right)
\lambda ,\theta \right) \approx 0.$,
while using a Taylor expansion gives:
$\sum_{j}\left( g\left( x_{j},\theta \right) +\partial_{z^{\prime }}g\left(
x_{j},\theta \right) \partial_{z}g^{\prime }\left( x_{j},\theta \right)
\lambda \right) \approx 0$, or:
\begin{equation*}
\hat{\mathbb{E}}\left[ g\left( x,\theta \right) \right] +\left( \hat{\mathbb{%
E}}\left[ H\left( x,\theta \right) H^{\prime }\left( x,\theta \right) \right]
\right) \lambda \approx 0
\end{equation*}%
where $H\left( x,\theta \right) =\partial_{z^{\prime }}g\left( x,\theta \right)$,
thus implying:%
\begin{equation}
\lambda \approx -\left( \hat{\mathbb{E}}\left[ H\left( x,\theta \right)
H^{\prime }\left( x,\theta \right) \right] \right) ^{-1}\hat{\mathbb{E}}%
\left[ g\left( x,\theta \right) \right]  \label{eqgivelambda}
\end{equation}%
Substituting (\ref{eqgivez}) and (\ref{eqgivelambda}) back into the
objective function (\ref{eqobjf}) yields (with $M\equiv \hat{\mathbb{E}}\left[ H\left( x,\theta \right) H^{\prime
}\left( x,\theta \right) \right]$) :%
\begin{eqnarray*}
\frac{1}{2n}\sum_{j}\left\Vert z_{j}-x_{j}\right\Vert ^{2}
&\approx &\frac{1}{2n}\sum_{j}\left\Vert x_{j}+H^{\prime }\left(
x_{j},\theta \right) \lambda -x_{j}\right\Vert ^{2} \\
\approx \frac{1}{2n}\sum_{j}\left\Vert -H^{\prime }\left( x_{j},\theta
\right) M ^{-1}\hat{\mathbb{E}}\left[ g\left(
x,\theta \right) \right] \right\Vert ^{2} &=&\frac{1}{2}\hat{\mathbb{E}}\left[ g^{\prime }\left( x,\theta \right) %
\right] M ^{-1}M M^{-1}\hat{\mathbb{E}}\left[ g\left(
x,\theta \right) \right] \\
&=&\frac{1}{2}\hat{\mathbb{E}}\left[ g^{\prime }\left( x,\theta \right) %
\right] M^{-1}\hat{\mathbb{E}}\left[ g\left(
x,\theta \right) \right]
\end{eqnarray*}%

\fixqed
\end{proof}

\subsection{Asymptotics} \label{appproof}

\begin{proof}[Proof of Theorem \protect\ref{thconsissml}]
We minimize $\frac{1}{2}\sum_{i=1}^{n}\Vert z_{i}-x_{i}\Vert ^{2}$ subject to $%
\sum_{i=1}^{n}g(z_{i},\theta )=0$. First-order conditions for $z_i$ read
\begin{equation}
z_{i}-x_{i}=\partial_{z}g^{\prime }(z_{i};\theta )\lambda  \label{FOC1}\\
\end{equation}
It is first shown that there exists a sequence $z_{i}^{\ast }$ that matches
the moment condition $\sum_{i=1}^{n}g(z_{i}^{\ast };\theta _{0})=0$ and
converges uniformly to the $x_{i}$'s, implying convergence of the $z_{i}$'s
by their definition in the optimization problem.

We now discuss how to eliminate observations that are too close to a zero gradient.
For some $\eta $ and $\delta $ let $A$ be the set of all $x_{i}$ such that $%
\inf_{y\in B_{\delta }(x_{i})}\Vert \partial_{z}g(y;\theta _{0})\Vert
\geq \eta $. We must have $\mathbb{P}[A]>0$ for some $(\eta ,\delta )$
because otherwise $\{\inf_{y\in B_{\delta }(x_{i})}\Vert \partial_{z}g(y;\theta _{0})\Vert \geq 1/n\}$ has probability 0 for all $n$,
thus $\{\inf_{y\in B_{\delta }(x_{i})}\Vert \partial_{z}g(y;\theta
_{0})\Vert >0\}$ has probability 0 for all $\delta $ by continuity from
below, contradicting assumption \ref{ass3sml} with continuity of $\partial_{z}g$.

We now consider such a pair $(\eta ,\delta )$, fix the resulting set $A$,
and let $A_{s}$ be the observations in sample that fall in it. 
In order to get enough degrees of freedom to offset deviations of sample
averages from 0, we make group of observations. Let $M\equiv
\dimop(g(z_{i};\theta _{0}))/\dimop(z_{i})$, and assume for convenience it is an
integer that divides $n-|A_{s}^{c}|$\footnote{%
If not, it suffices to set the remaining (components of) $z_{i}^{*}$ to $%
x_{i}$ and re-scale appropriately in what follows.}. Without loss, let the $%
x_{i}$ in $A_{s}^{c}$ constitute the first $|A_{s}^{c}|$ observations and
let $z_{i}^{\ast }=x_{i}$ for all $x_{i}\in A_{s}^{c}$. Then, for all $k\in 
\mathbb{N}$ (0 included) let $m_{k}\equiv \{|A_{s}^{c}|+Mk,\cdots
,|A_{s}^{c}|+Mk+M-1\}$ and solve wpa1 for $z_{i}^{\ast }$ in $\sum_{i\in
m_{k}}(g(z_{i}^{\ast };\theta _{0})-g(x_{i};\theta ))=-M\frac{n}{|A_{s}|}%
\frac{1}{n}\sum_{i=1}^{n}g(x_{i};\theta _{0})$. By the LLN $\frac{1}{n}%
\sum_{i=1}^{n}g(x_{i};\theta _{0})\rightarrow ^{p}0$ and $%
|A_{s}|/n\rightarrow ^{p}\mathbb{P}[A]>0$ so that a sequence $z_{i}^{\ast }$
with $z_{i}^{\ast }\rightarrow ^{p}x_{i}$ will exist by continuity.

We also get $\sup_{i}\Vert z_{i}^{\ast }-x_{i}\Vert \rightarrow ^{p}0$
because $\sup_{i}\Vert z_{i}^{\ast }-x_{i}\Vert \leq \frac{\sup_{i}\Vert
g(z_{i}^{\ast };\theta )-g(x_{i};\theta )\Vert }{\inf_{y\in A}\Vert \partial_{z}g(y;\theta )\Vert }\leq \eta o_{p}(1)$.
By definition of $z_{i}$ and the previous result, we have $\frac{1}{n}%
\sum_{i=1}^{n}\Vert z_{i}-x_{i}\Vert ^{2}\leq \frac{1}{n}\sum_{i=1}^{n}\Vert z_{i}^{\ast
}-x_{i}\Vert^{2}\leq \sup_{i}\Vert z_{i}^{\ast }-x_{i}\Vert^{2}\rightarrow ^{p}0$.
By properties of norms, assumptions \ref{ass3sml}-\ref{ass4sml} with H\"{o}lder continuity exponent $\alpha \leq 1$,
Cauchy-Schwartz, the LLN, and the previous convergence result
\begin{align*}
\begin{split}
&\left\Vert \frac{1}{n} \sum_{i=1}^n \partial_{z} g(z_i; \theta_0) (z_i - x_i) \right\Vert \\
&\leq \frac{1}{n} \sum_{i=1}^n \Vert \partial_{z} g(z_i;
\theta_0) - \partial_{z} g(x_i; \theta_0)\Vert \Vert z_i - x_i \Vert +\frac{1}{n} \sum_{i=1}^n \Vert \partial_{z} g(x_i; \theta_0)\Vert \Vert z_i
- x_i \Vert \\
& \leq C \frac{1}{n} \sum_{i=1}^n \Vert z_i - x_i \Vert^{1+\alpha} + \left(\frac{1}{n} \sum_{i=1}^n \Vert \partial_{z} g(x_i;
\theta_0)\Vert^2\right)^{1/2} \left(\frac{1}{n} \sum_{i=1}^n \Vert z_i -
x_i\Vert^2\right)^{1/2} \overset{p}{\rightarrow } 0
\end{split}%
\end{align*}
Furthermore, proceeding component-wise with $(k\cdot )$ denoting the $k^{%
\text{th}}$ row of a matrix and using assumptions \ref{ass3sml}-\ref{ass4sml} together with
previous results and proceeding as above for the term $\frac{1}{n}%
\sum_{i=1}^{n}\left\Vert [\partial_{z}g(z_{i};\theta _{0})]_{k\cdot
}\right\Vert \left\Vert z_{i}-x_{i}\right\Vert ^{\alpha }$, we have
\begin{align*}
\begin{split}
& \left\Vert \frac{1}{n} \sum_{i=1}^n [\partial_{z} g(z_i; \theta_0)]_{k
\cdot} [\partial_{z} g(z_i; \theta_0)]_{j \cdot}^{\prime}- \frac{1}{n}
\sum_{i=1}^n [\partial_{z} g(x_i; \theta_0)]_{k \cdot} [\partial_{z} g(x_i;
\theta_0)]_{j \cdot}^{\prime}\right\Vert \\
& \leq C \frac{1}{n} \sum_{i=1}^n \left\Vert [\partial_{z} g(z_i;
\theta_0)]_{k \cdot} \right\Vert \left\Vert z_i - x_i \right\Vert^\alpha
 + C \left(\frac{1}{n} \sum_{i=1}^n \left\Vert z_i - x_i
\right\Vert^{2\alpha}\right)^{1/2} \left(\frac{1}{n} \sum_{i=1}^n
\left\Vert[\partial_{z} g(x_i; \theta_0)]_{k \cdot} \right\Vert^2 \right)^{1/2}
\\
& \overset{p}{\rightarrow } 0 + \left(\mathbb{E}[\left\Vert[\partial_{z} g(x_i;
\theta_0)]_{k \cdot} \right\Vert^2]\right)^{1/2} 0 = 0
\end{split}%
\end{align*}
and thus, using assumption \ref{ass7sml}, there exists $C_n=O_p(1)$ such that
\begin{align*}
\Vert \lambda \Vert 
 \leq C_n \left\Vert \frac{1}{n} \sum_{i=1}^n \partial_{z} g(z_i; \theta_0) (z_i - x_i) \right\Vert \overset{p}{\rightarrow }  0
\end{align*}
Now we derive a precise rate of convergence and the resulting asymptotic distribution for $\lambda$.
Solving for $z_i$ in equation (\ref{FOC1}) yields $z_i(\lambda)$, which can be plugged in the second equation to obtain
$\sum_{i=1}^n g(z_i(\lambda); \theta) = 0.$ By a Taylor expansion and assumption \ref{ass7sml}, we get
\begin{equation}  \label{lambda_distrib}
\frac{1}{n} \sum_{i=1}^n g(x_i, \theta_0) + \frac{1}{n} \sum_{i=1}^n
\partial_{z} g(x_i; \theta_0) \partial_{z} g^{\prime}(x_i; \theta_0) \lambda +
O(\Vert\lambda\Vert^2) = 0
\end{equation}
By assumptions \ref{assiid}, \ref{ass1sml}, \ref{ass2sml} and the central limit theorem,
the first term is $O_{p}(n^{-1/2})$.

Under assumptions \ref{assiid} and \ref{ass7sml}, we have $\frac{1}{n}%
\sum_{i=1}^{n}\partial_{z}g(x_{i};\theta_0 )\partial_{z}g(x_{i};\theta_0
)^{\prime }\rightarrow ^{p}\linebreak \mathbb{E}[\partial_{z}g(x_{i};\theta_0 )
\partial_{z}g(x_{i};\theta_0 )^{\prime }]$ by the LLN and thus the second term is $%
O(\lambda )$. It follows that $\lambda =O_{p}(n^{1/2})$ with an
asymptotically normal distribution.

Finally, we turn to the situation where $\theta \neq \theta _{0}$. By the
uniform Law of Large Numbers, using assumption \ref{ass5sml}, $\sup_{\theta
\in \Theta }\Vert \frac{1}{n}\sum_{i=1}^{n}\allowbreak g(x_{i};\theta )-%
\mathbb{E}[g(x_{i};\theta )]\Vert \rightarrow ^{p}0$.

For any $\theta \in B_{\varepsilon }^{c}(\theta _{0})$ we have by
identification $\mathbb{E}[g(x_{i};\theta )]\in B_{\gamma }^{c}(0)$ for some 
$\gamma $ (otherwise, we can find a sequence whose mapping converges to 0
and by compactness there would be a convergent subsequence, implying
existence of some $\theta ^{\ast }\neq \theta _{0}$ that satisfies $\mathbb{E%
}[g(x_{i};\theta ^{\ast })]=0$).

With probability approaching one, we have by the mean value theorem and
Cauchy-Schwartz $\frac{\gamma }{2}\leq \frac{1}{n}\sum_{i=1}^{n}\left\Vert
g(z_{i};\theta )-g(x_{i};\theta )\right\Vert =\frac{1}{n}\sum_{i=1}^{n}\left%
\Vert g(\overline{z}_{i};\theta )(z_{i}-x_{i})\right\Vert\linebreak \leq (\frac{1}{n}%
\sum_{i=1}^{n}\left\Vert g(\overline{z}_{i};\theta )\right\Vert
^{2})^{1/2}\allowbreak (\frac{1}{n}\sum_{i=1}^{n}\left\Vert
z_{i}-x_{i}\right\Vert ^{2})^{1/2}$. As a result, $\frac{1}{n}%
\sum_{i=1}^{n}\left\Vert z_{i}-x_{i}\right\Vert ^{2}\rightarrow ^{p}0$ (or a
subsequence) would imply $\frac{1}{n}\sum_{i=1}^{n}\left\Vert g(\overline{z}%
_{i};\theta )\right\Vert ^{2}\rightarrow ^{p}\mathbb{E}[\Vert g(x_{i};\theta
)\Vert^2 ]$ as before and thus $\gamma \leq o_{p}(1)$, which is impossible.
Therefore, $\sum_{i=1}^{n}\Vert z_{i}-x_{i}\Vert ^{2}>O(n)$ with probability
approaching one, and the probability that $\hat{\theta}$ lives outside any
neighborhood of $\theta _{0}$ decreases to 0.

Eventually, the first-order conditions read $z_{i}-x_{i}=\lambda ^{\prime
}\partial_{z}g(x_{i};\theta _{0})+o_{p}(n^{-1/2})$ and 
$\frac{1}{n}\sum_{i=1}^{n}g(z_{i};\theta _{0})=0$ and the linearized version
is asymptotically justified.
\end{proof}

\begin{proof}[Proof of Theorem \protect\ref{thnormsml}]
Using first-order conditions, previous results, and equation (\ref{lambda_distrib}), we have
\begin{align*}
F& \equiv \frac{1}{n}\sum_{i=1}^{n}\Vert z_{i}-x_{i}\Vert ^{2} =\lambda ^{\prime }\frac{1}{n}\sum_{i=1}^{n}\partial_{z}g(z_{i};\theta
)\partial_{z}g(z_{i};\theta )^{\prime }\lambda  =\lambda ^{\prime }\frac{1}{n}\sum_{i=1}^{n}\partial_{z}g(x_{i};\theta
)\partial_{z}g(x_{i};\theta )^{\prime }\lambda +o_{p}(F) \\
& =\left( \left( \frac{1}{n}\sum_{i=1}^{n}\partial_{z}g(x_{i};\theta
)\partial_{z}g(x_{i};\theta )^{\prime }\right) ^{-1}\frac{1}{n}%
\sum_{i=1}^{n}g(x_{i};\theta )\right) ^{\prime }\frac{1}{n}%
\sum_{i=1}^{n}\partial_{z}g(x_{i};\theta )\partial_{z}g(x_{i};\theta
)^{\prime } \\
& \left( \left( \frac{1}{n}\sum_{i=1}^{n}\partial_{z}g(x_{i};\theta
)\partial_{z}g(x_{i};\theta )^{\prime }\right) ^{-1}\frac{1}{n}%
\sum_{i=1}^{n}g(x_{i};\theta )\right) +O_{p}(\Vert \lambda \Vert
^{3})+O(\Vert \lambda \Vert ^{4})+o_{p}(F)
\end{align*}
Ignoring lower order terms, we can eventually reframe the problem as
minimizing standard GMM: 
$\sum_{i=1}^{n}g(x_{i};\theta )^{\prime }(\sum_{i=1}^{n}\partial_{z}g(x_{i};\theta _{0})\partial_{z}g^{\prime }(x_{i},\theta
_{0}))^{-1}\sum_{i=1}^{n}g(x_{i};\theta )$ to get the first-order conditions%
\begin{equation*}
\sum_{i=1}^{n}\partial_{\theta}g(x_{i};\theta )^{\prime }\left(
\sum_{i=1}^{n}\partial_{z}g(x_{i};\theta _{0})\partial_{z}g^{\prime
}(x_{i},\theta _{0})\right) ^{-1}\sum_{i=1}^{n}g(x_{i};\theta )=0
\end{equation*}%
which are satisfied with probability approaching 1. By an expansion around $\theta _{0}$,%
\begin{equation*}
\sum_{i=1}^{n}\partial_{\theta}g(x_{i};\theta )^{\prime }\left(
\sum_{i=1}^{n}\partial_{z}g(x_{i};\theta _{0})\partial_{z}g^{\prime
}(x_{i},\theta _{0})\right) ^{-1}\sum_{i=1}^{n}[g(x_{i};\theta
_{0})+\partial_{\theta}g(x_{i};\overline{\theta })(\theta -\theta _{0})]=0
\end{equation*}%
so that the estimator takes the form%
\begin{equation*}
\begin{split}
\hat{\theta}_{OTGMM}-\theta _{0}& =-\left( \sum_{i=1}^{n}\partial_{\theta}g(x_{i};\hat{\theta})^{\prime }\left( \sum_{i=1}^{n}\partial_{z}g(x_{i};\theta _{0})\partial_{z}g^{\prime }(x_{i},\theta _{0})\right)
^{-1}\sum_{i=1}^{n}\partial_{\theta}g(x_{i};\overline{\theta })\right) ^{-1} \\
& \left( \sum_{i=1}^{n}\partial_{\theta}g(x_{i};\hat{\theta})^{\prime }\left(
\sum_{i=1}^{n}\partial_{z}g(x_{i};\theta _{0})\partial_{z}g^{\prime
}(x_{i},\theta _{0})\right) ^{-1}\sum_{i=1}^{n}g(x_{i};\theta _{0})\right)
\end{split}%
\end{equation*}
Noting that under assumptions \ref{assiid}, \ref{ass1sml}, and \ref{ass2sml} 
$\frac{1}{\sqrt{n}}\sum_{i=1}^{n}g(x_{i};\theta_0 )$ converges in distribution
to a normal random variables by the central limit theorem and that
assumptions \ref{ass8sml}-\ref{ass10sml} together with consistency
ensure convergence of sample averages to expectations, we obtain the
asymptotic normality of $\sqrt{n}(\hat{\theta}_{OTGMM}-\theta _{0})$ by
Slutsky with asymptotic variance given in the theorem.
\end{proof}

\begin{proof}[Proof of Theorem \protect\ref{thjidgmm}]
The first-order conditions with respect to the $z_{j}$ (Equation (\ref%
{eqfoczj})) can be written as
\begin{equation}
x_{j}=z_{j}-\partial_{z}g^{\prime }\left( z_{j},\theta \right) \lambda .
\label{eqfoclam}
\end{equation}
Under our assumptions, Equation (\ref{eqfoclam}) defines a direct
relationship between $z_{j}$ and $x_{j}$, and therefore an implicit reverse
relationship between $x_{j}$ and $z_{j}$. Since the latter may not be
unique, we observe that our original optimization problem seeks to minimize
the distance between $x_{j}$ and $z_{j}$. Hence, in cases where (\ref%
{eqfoclam}) admits multiple solutions $z_{j}$ for a given $x_{j}$, we
identify the unique (with probability one) solution that minimizes $%
\left\Vert z_{j}-x_{j}\right\Vert ^{2}$. This is accomplished by defining
the mapping (\ref{eqinvfoclam}). With this definition, the first-order conditions (\ref{eqfoctheta}) and (\ref%
{eqfoclambda}) of the Lagrangian optimization problem for $\theta $ and $%
\lambda $ yield the just-identified GMM estimator stated in the Theorem.
\end{proof}

The following Lemmas are shown in Supplement \ref{secprooflem}.

\begin{lemma}
\label{lemdiffq}Let $h\left( \cdot ,\cdot ,\cdot \right) $ be continuous in
all of its arguments. Then, under Assumptions \ref{assgcontdiff}(i) and \ref%
{asseigval}, $h\left( q\left( x,\theta ,\lambda \right) ,\theta ,\lambda
\right) $ is continuous in $\left( \theta ,\lambda \right) $.
\end{lemma}

\begin{lemma}
\label{lemholderLD}Under Assumptions \ref{assidgmm} and \ref{asseigval}, if $%
h\in \mathcal{L}$, then, for $q\left( x,\theta ,\lambda \right) $ defined in Theorem \ref{thjidgmm},  
$\mathbb{E}\left[ \sup_{\left( \theta ,\lambda \right) \in \Theta \times
\Lambda }\left\Vert h\left( q\left( x,\theta ,\lambda \right) ,\theta
\right) \right\Vert \right] <\infty$.
\end{lemma}

\begin{proof}[Proof of Theorem \protect\ref{thconsislarge}]
Assumptions \ref{assiid}-\ref{asssimplegtilda} directly imply consistency of
our GMM estimator, by Theorem 2.6 in \citet{Newey:HB}. There remains to show
that Assumption \ref{asssimplegtilda} is implied by Assumptions \ref%
{assgcontdiff},\ref{asseigval},\ref{assgdom}.

We first establish Assumption \ref{asssimplegtilda} (i): Continuity of $%
\tilde{g}\left( x,\theta ,\lambda \right) $ in $\left( \theta ,\lambda
\right) $. To show that $g\left( q\left( x,\theta ,\lambda \right) ,\theta
\right) $ is continuous in $\left( \theta ,\lambda \right) $, we can invoke
Lemma \ref{lemdiffq} for $h\left( z,\theta ,\lambda \right) =g\left(
z,\theta \right) $, under Assumptions \ref{assgcontdiff}(i) and \ref%
{asseigval}. To show that $\partial_{\theta}g^{\prime }\left( q\left( x,\theta
,\lambda \right) ,\theta \right) \lambda $ is continuous in $\left( \theta
,\lambda \right) $, we can similarly invoke Lemma \ref{lemdiffq} for $%
h\left( z,\theta ,\lambda \right) =\partial_{\theta}g^{\prime }\left( z,\theta
\right) \lambda $, where $\partial_{\theta}g^{\prime }\left( z,\theta \right) $
is continuous in both arguments by Assumption \ref{assgcontdiff}(ii).

We now establish Assumption \ref{asssimplegtilda} (ii): $\mathbb{E}\left[
\sup_{\left( \theta ,\lambda \right) \in \Theta \times \Lambda }\left\Vert 
\tilde{g}\left( x,\theta ,\lambda \right) \right\Vert \right] <\infty $.
Since $g\left( \cdot ,\cdot \right) \in \mathcal{L}$ by Assumption \ref%
{assgdom}, it follows that $\mathbb{E}\left[ \sup_{\left( \theta ,\lambda
\right) \in \Theta \times \Lambda }\left\Vert g\left( q\left( x,\theta
,\lambda \right) ,\theta \right) \right\Vert \right] <\infty $, by Lemma \ref%
{lemholderLD}. Next, we have, for $\left( \theta ,\lambda \right) \in \Theta
\times \Lambda $, $\left\Vert \partial_{\theta}g^{\prime }\left( q\left(
x,\theta ,\lambda \right) ,\theta \right) ~\lambda \right\Vert \leq
\left\Vert \partial_{\theta}g^{\prime }\left( q\left( x,\theta ,\lambda \right)
,\theta \right) \right\Vert \left\Vert \lambda \right\Vert \leq \left\Vert
\partial_{\theta}g^{\prime }\left( q\left( x,\theta ,\lambda \right) ,\theta
\right) \right\Vert \bar{\lambda}$ by Assumption \ref{asseigval} and
compactness of $\Lambda $. By Assumption \ref{assgdom} and Lemma \ref%
{lemholderLD} we then also have that\linebreak $\mathbb{E}\left[ \sup_{\left( \theta
,\lambda \right) \in \Theta \times \Lambda }\left\Vert \partial_{\theta}g^{\prime }\left( q\left( x,\theta ,\lambda \right) ,\theta \right)
~\lambda \right\Vert \right] <\infty $.
\end{proof}

\begin{proof}[Proof of Theorem \protect\ref{thanlarge}]
Theorem \ref{thconsislarge} implies consistency $\left( \theta ,\lambda
\right) \overset{p}{\longrightarrow }\left( \theta _{0},\lambda _{0}\right) $%
. This, in addition to Assumptions \ref{assinterlarge}, \ref{asshesjaclarge}
and \ref{assdgdom} directly implies the stated asymptotic normality result,
by Theorem 3.2 and Lemma 2.4 in \citet{Newey:HB} and the Lindeberg-Levy
Central Limit Theorem.\ There remains to show that Assumption \ref{assdgdom}
is implied by Assumption \ref{assdgdomprim}.

By Lemma \ref{lemdiffq}, Assumptions \ref{assdgdomprim}(i) and (iii) imply
that both $g\left( q\left( x,\theta ,\lambda \right) ,\lambda \right) $ and $%
\partial_{\theta}g^{\prime }\left( q\left( x,\theta ,\lambda \right) ,\lambda
\right) \lambda $ are continuously differentiable in $\left( \theta ,\lambda
\right) $, thus establishing Assumption \ref{assdgdom}(i).
By Lemma \ref{lemholderLD}, Assumptions \ref{assidgmm}, \ref{assdgdomprim}%
(ii) and (iii) imply Assumption \ref{assdgdom}(ii).

The asymptotic variance of the just-identified GMM estimator defined in
Theorem \ref{thjidgmm}\ is then given by $\left( \tilde{G}^{\prime }\Omega
^{-1}\tilde{G}\right) ^{-1}$ for $\Omega$ and $\tilde{G}$ as defined in the Theorem statement.

Finally, the explicit expressions for the derivatives of the function $%
z=q\left( x,\theta ,\lambda \right) $ follow from the implicit function
theorem after noting that $q\left( x,\theta ,\lambda \right) $ is the
inverse of the mapping $z\mapsto z-\partial_{z}\left( \lambda ^{\prime
}g\left( z,\theta \right) \right) $.\ This can also be shown through an
explicit calculation: To first order, (\ref{eqfoclam}) implies, for a small
change $\Delta \theta $ in $\theta $, a corresponding change $\Delta z$ in $%
z $ while keeping $x$ and $\lambda $ fixed, that:%
\begin{equation*}
0=\Delta z-\partial_{zz^{\prime }}\left( \lambda ^{\prime }g\left( z,\theta
\right) \right) \Delta z-\partial_{z\theta^{\prime }}\left( \lambda ^{\prime
}g\left( z,\theta \right) \right) \Delta \theta .
\end{equation*}%
Thus, $\Delta z=\left( I-\partial_{zz^{\prime }}\left( \lambda ^{\prime }g\left(
z,\theta \right) \right) \right) ^{-1}\partial_{z\theta^{\prime }}\left( \lambda
^{\prime }g\left( z,\theta \right) \right) \Delta \theta$ and we have:%
\begin{equation*}
\partial_{\theta^{\prime }}q\left( x,\theta ,\lambda \right) =\left( I-\partial_{zz^{\prime }}\left( \lambda ^{\prime }g\left( z,\theta \right) \right)
\right) ^{-1}\partial_{z\theta^{\prime }}\left( \lambda ^{\prime }g\left(
z,\theta \right) \right)  \label{eqqt}
\end{equation*}%
evaluated at $z=q\left( x,\theta ,\lambda \right) $. A similar reasoning for $\lambda $, exploiting the fact that $\frac{\partial ^{2}\left( \lambda
^{\prime }g\left( z,\theta \right) \right) }{\partial z\partial \lambda
^{\prime }}=\frac{\partial g^{\prime }\left( z,\theta \right) }{\partial z}$%
, yields:
\begin{equation*}
\partial_{\lambda^{\prime }}q\left( x,\theta ,\lambda \right) =\left( I-\partial_{zz^{\prime }}\left( \lambda ^{\prime }g\left( z,\theta \right) \right)
\right) ^{-1}\partial_{z}g^{\prime }\left( z,\theta \right) .  \label{eqql}
\end{equation*}%

\fixqed
\end{proof}

\clearpage

\bibliographystyle{econometrica}
\bibliography{schennach}

@incollection{Newey:HB,
  author="W. Newey and D. McFadden",
  title="Large Sample Estimation and Hypothesis Testing",
  booktitle="Handbook of Econometrics",
  volume="IV",
  editor="R. F. Engel and D. L. McFadden",
  publisher="Elsevier Science",
  year=1994
}

@article{schennach:nlme,
	author="S. M. Schennach",
	title="Estimation of Nonlinear Models with Measurement Error",
	journal="Econometrica",
	volume="72",
	pages="33--75",
	year=2004
}

@article{chesher:1991,
	author="A. Chesher",
	title="The Effect of Measurement Error",
	journal="Biometrika",
	volume=78,
	pages="451",
	year=1991
}

@article{chamberlain:gmmeff,
	title="Asymptotic efficiency in estimation with conditional moment restrictions",
	author="G. Chamberlain",
	journal="Journal of Econometrics",
	volume=34,
	year="1987",
	pages="305--334"
}

@article{schennach:annrev,
	author="S. M. Schennach",
	title="Recent Advances in the Measurement Error Literature",
	volume="8",
	journal="Annual Reviews of Economics",
	pages="341--377",
	year=2016
}

@incollection{schennach:HB,
   author="S. M. Schennach",
   title="Mismeasured and unobserved variables",
   booktitle="Handbook of Econometrics",
   editor="S. Durlauf, L. Hansen, J. Heckman and R. Matzkin",
   chapter=6,
   pages="487--565",
   volume="7A",
   publisher="Elsevier Science",
   year=2020
}

@article{schennach:elvis,
	title="Entropic Latent Variable Integration via Simulation",
	author="S. M. Schennach",
	journal="Econometrica",
	volume="82",
	pages="345--386",
	year=2014
}

@article{aguiar:revpref,
	author="V. H. Aguiar and N. Kashaev",
	title="Stochastic Dynamic Revealed Preferences with measurement error",
	journal="Review of Economic Studies",
	volume=88,
	year=2021,
	pages="2042--2093"
}

@article{doraszelski:megmm,
  title="{R\&D} and Productivity: Estimating Endogenous Productivity",
  author="U. Doraszelski and J. Jaumandreu",
  journal="Review of Economic Studies",
  year=2013,
  volume=80,
  pages="1338--1383"
}

@article{owen:el1,
   author="A. B. Owen",
   title="Empirical Likelihood Ratio Confidence Intervals for a Single Functional",
   journal="Biometrika",
   volume=75,
   year=1988,
   pages="237--249"
}

@article{imbens:res,
   author="G. W. Imbens",
   title="One-Step Estimators for Over-Identified Generalized Method of Moments Models",
   journal="Review of Economic Studies",
   volume=64,
   year=1997,
   pages="359--383"
}

@article{kitamura:itgmm,
   author="Y. Kitamura and M. Stutzer",
   title="An Information-Theoretic Alternative to Generalized Method of Moment Estimation",
   journal="Econometrica",
   volume=65,
   year=1997,
   pages="861--874"
}

@article{qinlawless,
   author="J. Qin and J. Lawless",
   title="Empirical Likelihood and General Estimating Equations",
   journal="Annals of Statistics",
   volume=22,
   year=1994,
   pages="300--325"
}

@article{newey:gel,
	title="Higher-Order Properties of {GMM} and Generalized Empirical Likelihood
		Estimators",
	author="W. Newey and R. J. Smith",
	journal="Econometrica",
	volume=72,
	pages="219--255",
	year="2004"
}

@article{hansen:gmm,
	title="Large sample properties of generalized method of moment estimators",
	author="L. P. Hansen",
	journal="Econometrica",
	volume=50,
	pages="1029-1054",
	year="1982"
}

@article{white:miss,
	title="Maximum Likelihood Estimation of Misspecified Models",
	author="H. White",
	journal="Econometrica",
	volume=50,
	pages="1--26",
	year="1982"
}

@article{schennach:etel,
	author="S. M. Schennach",
	title="Point Estimation with Exponentially Tilted Empirical Likelihood",
	journal="Annals of Statistics",
	volume="35",
	year="2007",
	pages="634--672"
}

@article{hall:missgmm,
	title="The large sample behavior of the generalized method of moments estimator in misspecified models",
	author="A. R. Hall and A. Inoue",
	journal="Journal of Econometrics",
	volume="114",
	pages="361--394",
	year="2003"
}

@article{sawa:pseudo,
	title="Information Criteria for Discriminating Among Alternative Regression Models",
	author="T. Sawa",
	journal="Econometrica",
	volume="46",
	pages="1273--1291",
	year="1978"
}

@article{galichon:vectquan,
  title="Comonotonic Measures of Multivariate Risks",
  author="I. Ekeland and A. Galichon and M. Henry",
  journal="Mathematical Finance",
  volume=22,
  pages="109--132",
  year=2011
}

@incollection{villani:new,
	author="C. Villani",
	title="Optimal transport: Old and New",
	booktitle="Grundlehren der mathematischen Wissenschaften",
	publisher="Springer-Verlag",
	address="Heidelberg",
	year=2009
}

@book{galichon:optran,
	author="A. Galichon",
	year=2016,
	title="Optimal Transport Methods in Economics",
	publisher="Princeton University Press",
	address="Princeton"
}

@book{santambrogio:optran,
	author="F. Santambrogio",
	year=2015,
	title="Optimal Transport for Applied Mathematicians",
	publisher="Springer",
	address="New York"
}

@article{schennach:nlfact,
	author="F. Gunsilius and S. M. Schennach",
	title="Independent Principal Component Analysis",
	journal="Journal of the American Statistical Association",
	volume="forthcoming",
	year=2021
}

@article{cherno:breniercons,
  title="Monge–Kantorovich depth, quantiles, ranks and signs",
  author="V. Chernozhukov and A. Galichon and M. Hallin and M. Henry",
  journal="Annals of Statistics",
  volume=45,
  pages="223--256",
  year=2017
}

@article{carlier:vectq,
	author = {G. Carlier and V. Chernozhukov and A. Galichon},
	title = {Vector quantile regression: An optimal transport approach},
	journal = {The Annals of Statistics},
	volume = {44},
	pages = {1165--1192},
	year = 2016
}

@article{hansen:miss,
  title="Acknowledging Misspecification in Macroeconomic Theory",
  author="L. P. Hansen",
  journal="Review of Economic Dynamics",
  volume=4,
  pages="519--535",
  year=2001
}

@article{duranton2014roads,
  title={Roads and Trade: Evidence from the US},
  author={Duranton, Gilles and Morrow, Peter M and Turner, Matthew A},
  journal={Review of Economic Studies},
  volume={81},
  number={2},
  pages={681--724},
  year={2014},
  publisher={Oxford University Press}
}

@article{poirier:salvaging,
  title={Salvaging falsified instrumental variable models},
  author={Masten, Matthew A and Poirier, Alexandre},
  journal={Econometrica},
  volume={89},
  number={3},
  pages={1449--1469},
  year={2021},
  publisher={Wiley Online Library}
}

@techreport{kwon:ineqmiss,
  title="Inference in moment inequality models that is robust to spurious precision under model misspecification",
  author="D. W. K. Andrews and S. Kwon",
  institution="cowles foundation discussion paper",
  year=2019,
  number="2184R"
}

@article{shapiro:sens,
  title="Measuring the sensitivity of parameter estimates to estimation moments",
  author="I. Andrews and M. Gentzkow and J. M. Shapiro",
  journal="Quarterly journal of economics",
  volume=132,
  pages="1553--1592",
  year=2017
}

@article{conley:plausexo,
  author="T. G. Conley and C. B. Hansen and P. E. Rossi",
  year=2012,
  title="Plausibly Exogenous",
  journal="The Review of Economics and Statistics",
  volume=94,
  pages="260--272"
}

@article{connault:sens,
  title="Counterfactual Sensitivity and Robustness",
  author="T. Christensen and B. Connault",
  journal="Econometrica",
  volume="forthcoming",
  year=2022
}

@article{blanchet:distrob,
  title="Confidence regions inWasserstein distributionally robust estimation",
  author="J. Blanchet and K. Murthy and N. Si",
  journal="Biometrika",
  volume=109,
  pages="295--315",
  year=2022
}

@article{bonhomme:sens,
  title="Minimizing sensitivity to model misspecification",
  author="S. Bonhomme and M. Weidner",
  journal="Quantitative Economics",
  volume=13,
  year=2022,
  pages="907--954"
}

@article{caffarelli:regconv,
  title="Boundary Regularity of Maps with Convex Potentials--II",
  author="L. A. Caffarelli",
  journal="Annals of Mathematics, Second Series",
  volume=144,
  year=1996,
  pages="453--496"
}

\clearpage

\setcounter{page}{1}
\setcounter{section}{0}
\setcounter{table}{0}
\renewcommand{\thetable}{S.\arabic{table}}
\renewcommand{\thesection}{S.\arabic{section}}
\begin{center}
{\large Online Supplement for:\\Optimally-Transported Generalized Method of Moments\\}
by Susanne Schennach and Vincent Starck
\end{center}

\section{Algorithms}

\subsection{Iterative solution}

\label{appiter}The first order condition (\ref{eqfoczj}) can be re-written as%
\begin{equation}
\left( z_{j}-x_{j}\right) =\partial_{z}g^{\prime }\left( z_{j},\theta
\right) \lambda .  \label{eqfoczj2}
\end{equation}%
We seek to construct a sequence $z_{j}^{t}$ ($t=0,1,\ldots $) that converges
to $z_{j}$, starting with $z_{j}^{t}|_{t=0}=x_{j}$. From the moment
conditions and (\ref{eqfoczj2}), we have:%
\begin{equation*}
0=\frac{1}{n}\sum_{i=1}^{n}g\left( z_{i},\theta \right) =\frac{1}{n}%
\sum_{i=1}^{n}g\left( x_{j}+\partial_{z}g^{\prime }\left( z_{j},\theta
\right) \lambda ,\theta \right) .
\end{equation*}%
Adding and subtracting $z_{j}^{t}$ yields 
\begin{eqnarray*}
0 &=&\frac{1}{n}\sum_{i=1}^{n}g\left( z_{j}^{t}+\left(
x_{j}-z_{j}^{t}+\partial_{z}g^{\prime }\left( z_{j},\theta \right) \lambda
\right) ,\theta \right) \\
&\approx &\frac{1}{n}\sum_{i=1}^{n}g\left( z_{j}^{t},\theta \right) +\frac{1%
}{n}\sum_{i=1}^{n}\partial_{z^{\prime }}g\left( z_{j}^{t},\theta \right)
\left( x_{j}-z_{j}^{t}+\partial_{z}g^{\prime }\left( z_{j},\theta \right)
\lambda \right)
\end{eqnarray*}%
where the expansion is justified from the fact that $x_{j}-z_{j}^{t}+%
\partial_{z}g^{\prime }\left( z_{j},\theta \right) \lambda \longrightarrow
0 $ as $z_{j}^{t}\longrightarrow z_{j}$.

In the same limit, $\partial_{z}g^{\prime }\left( z_{j}^{t},\theta \right)
\longrightarrow \partial_{z}g^{\prime }\left( z_{j},\theta \right) $, so 
\begin{eqnarray*}
0 &\approx &\frac{1}{n}\sum_{i=1}^{n}g\left( z_{j}^{t},\theta \right) +\frac{%
1}{n}\sum_{i=1}^{n}\partial_{z^{\prime }}g\left( z_{j}^{t},\theta \right)
\left( x_{j}-z_{j}^{t}+\partial_{z}g^{\prime }\left( z_{j}^{t},\theta
\right) \lambda \right) \\
&=&\frac{1}{n}\sum_{i=1}^{n}g\left( z_{j}^{t},\theta \right) +\frac{1}{n}%
\sum_{i=1}^{n}\partial_{z^{\prime }}g\left( z_{j}^{t},\theta \right) \left(
x_{j}-z_{j}^{t}\right) +\frac{1}{n}\sum_{i=1}^{n}\partial_{z^{\prime}}g\left( z_{j}^{t},\theta \right) \partial_{z}g^{\prime }\left(
z_{j}^{t},\theta \right) \lambda \\
&=&\hat{\mathbb{E}}\left[ g\left( z^{t},\theta \right) \right] +\hat{\mathbb{%
E}}\left[ H\left( z^{t},\theta \right) \left( x-z^{t}\right) \right] +\hat{%
\mathbb{E}}\left[ H\left( z^{t},\theta \right) H^{\prime }\left(
z^{t},\theta \right) \right] \lambda .
\end{eqnarray*}%
Isolating $\lambda $ gives the approximation to the Lagrange multiplier at
step $t+1$:%
\begin{equation}
\lambda ^{t+1}=\left( \hat{\mathbb{E}}\left[ H\left( z^{t},\theta \right)
H^{\prime }\left( z^{t},\theta \right) \right] \right) ^{-1}\left( -\hat{%
\mathbb{E}}\left[ g\left( z^{t},\theta \right) \right] +\hat{\mathbb{E}}%
\left[ H\left( z^{t},\theta \right) \left( z^{t}-x\right) \right] \right) .
\label{eqlamtp1}
\end{equation}%
From this, we can improve the approximation to $z_{j}$ to go to the next
step, using (\ref{eqfoczj2}):%
\begin{equation}
z_{j}^{t+1}=x_{j}+H^{\prime }\left( z^{t},\theta \right) \lambda ^{t+1}\text{%
.}  \label{eqztp1}
\end{equation}%
It can be directly verified that the values of $z_{j}$ and $\lambda $ that
satisfy the first order conditions are indeed a fixed point of this
iterative rule. In the next subsection we shall provide conditions under
which this fixed point is also attractive.

After iteration to convergence, the objective function can be written in
term of the converged values of $z$ and $\lambda $:%
\begin{eqnarray*}
\hat{Q}(\theta ) &=&\frac{1}{2n}\sum_{j}\left\Vert z_{j}-x_{j}\right\Vert
^{2}=\frac{1}{2n}\sum_{j}\left\Vert H^{\prime }\left( z_{j},\theta \right)
\lambda \right\Vert ^{2}=\frac{1}{2n}\sum_{j}\lambda ^{\prime }H\left(
z_{j},\theta \right) H^{\prime }\left( z_{j},\theta \right) \lambda \\
&=&\frac{1}{2}\lambda ^{\prime }\hat{\mathbb{E}}\left[ H\left( z,\theta
\right) H^{\prime }\left( z,\theta \right) \right] \lambda \text{.}
\end{eqnarray*}

\subsection{Iterative procedure convergence}

\label{appiterconv}Substituting (\ref{eqlamtp1}) into (\ref{eqztp1}) yields
an iterative rule expressed solely in terms of $z_{j}^{t}$:%
\begin{equation}
\footnotesize
z_{j}^{t+1}=x_{j}+H^{\prime }\left( z_{j}^{t},\theta \right) \left( \hat{%
\mathbb{E}}\left[ H\left( z^{t},\theta \right) H^{\prime }\left(
z^{t},\theta \right) \right] \right) ^{-1}\left( -\hat{\mathbb{E}}\left[
g\left( z^{t},\theta \right) \right] +\hat{\mathbb{E}}\left[ H\left(
z^{t},\theta \right) \left( z^{t}-x\right) \right] \right) .
\label{eqztonly}
\end{equation}%
This is an iterative rule of the form $\mathbf{z}^{t+1}=f\left( \mathbf{z}%
^{t}\right) $, for $\mathbf{z}^{t}=\left( z_{1}^{t\prime },\ldots
,z_{n}^{t\prime }\right) ^{\prime }\in \mathbb{R}^{nd_{x}}$ with fixed point
denoted $\mathbf{z}^{\infty }$. We then have the following result.

\begin{condition}
\label{assitconv}(i) $g\left( z,\theta \right) $ is twice continuously
differentiable in $z$ and\linebreak (ii) $\hat{\mathbb{E}}\left[ H\left( z,\theta
\right) H^{\prime }\left( z,\theta \right) \right] $ is nonsingular for $%
\mathbf{z}$ in the closure of an open neighborhood $\eta $ of the fixed
point $\mathbf{z}^{\infty }$.
\end{condition}

\begin{theorem}
\label{thitconv}Under Assumption \ref{assitconv}, for a given sample $%
x_{1},\ldots ,x_{n}$, there exists a neighborhood $\eta $ of $\mathbf{z}%
^{\infty }$, such that the iterative procedure defined by Equation (\ref%
{eqztonly}) and starting at any $\mathbf{z}^{0}\in \eta $ converges to the
unique fixed point $\mathbf{z}^{\infty }\in \eta$, provided $\left\Vert \lambda \right\Vert 
$ is sufficiently small (where $\lambda $ solves the first order condition (%
\ref{eqfoczj2})).
\end{theorem}

The condition that the initial point $\mathbf{z}^{0}$ should lie in a
neighborhood of the solution is standard --- most Newton-Raphson-type
iterative refinements have a similar requirement. If necessary, this
requirement can be met by simply attempting many different starting points
in search for one that yields a convergent sequence. The condition that $%
\lambda $ be small intuitively means that the errors should not
be too large. This is a purely numerical condition which has nothing to do
with sample size, statistical significance of specification tests. In
particular, it does not mean that the error magnitude must
decreases with sample size or that the effect of the errors should be small
relative to the estimator's standard deviation. Typically, the constraint on 
$\lambda $ is relaxed as the starting point $\mathbf{z}^{0}$ is chosen
closer to the solution $\mathbf{z}^{\infty }$.

\begin{proof}[Proof of Theorem \protect\ref{thitconv}]
For a rule of the form $\mathbf{z}^{t+1}=f\left( \mathbf{z}^{t}\right) $,
Banach's Fixed Point Theorem applied to a neighborhood of $\mathbf{z}%
^{\infty }$ provides a simple sufficient condition for convergence: (i) $f$
must be continuously differentiable in a neighborhood of $\mathbf{z}^{\infty
}$ and (ii) all eigenvalues of the matrix $\left[ \partial f\left( \mathbf{z}%
\right) /\partial \mathbf{z}^{\prime }\right] _{\mathbf{z}=\mathbf{z}%
^{\infty }}$ must have a modulus strictly less than $1$.

The smoothness condition (i) is trivially satisfied under Assumption \ref%
{assitconv}. Next, letting $z_{i,k}^{t}$ denote one element of the vector $%
z_{i}^{t}$, and $H_{\cdot k}\left( z_{i}^{t},\theta \right) $ denote the $k^{%
\text{th}}$ column of $H\left( z_{i}^{t},\theta \right) $, we can express
all blocks $\partial z_{j}^{t+1}/\partial z_{i,k}^{t}$ of the matrix of
partial derivatives of $f\left( \mathbf{z}\right) :$%
\begin{eqnarray*}
\frac{\partial z_{j}^{t+1}}{\partial z_{i,k}^{t}} &=&\left[ \frac{\partial }{%
\partial z_{i,k}^{t}}H^{\prime }\left( z_{i}^{t},\theta \right) \left( \hat{%
\mathbb{E}}\left[ H\left( z^{t},\theta \right) H^{\prime }\left(
z^{t},\theta \right) \right] \right) ^{-1}\right]\times\\
&&\left( -\hat{\mathbb{E}}%
\left[ g\left( z^{t},\theta \right) \right] +\hat{\mathbb{E}}\left[ H\left(
z^{t},\theta \right) \left( z^{t}-x\right) \right] \right)+ \\
&&H^{\prime }\left( z_{i}^{t},\theta \right) \left( \hat{\mathbb{E}}\left[
H\left( z^{t},\theta \right) H^{\prime }\left( z^{t},\theta \right) \right]
\right) ^{-1}\times \\
&&\left( -n^{-1}H_{\cdot k}\left( z_{i}^{t},\theta \right) +n^{-1}H_{\cdot
k}\left( z_{i}^{t},\theta \right) +n^{-1}\left[ \frac{\partial }{\partial
z_{i,k}^{t}}H\left( z_{i}^{t},\theta \right) \right] \left(
z_{i}^{t}-x_{i}\right) \right) ,
\end{eqnarray*}%
where the two $n^{-1}H_{\cdot k}\left( z_{i}^{t},\theta \right) $ terms
cancel each other. At $\mathbf{z}^{t}=\mathbf{z}^{\infty }$, $\hat{\mathbb{E}%
}\left[ g\left( z^{\infty },\theta \right) \right] =0$ and $\left(
z_{i}^{\infty }-x_{i}\right) =H^{\prime }\left( z_{i}^{\infty },\theta
\right) \lambda $ and we have:%
\begin{eqnarray*}
\frac{\partial z_{j}^{t+1}}{\partial z_{i,k}^{t}} &=&\left[ \frac{\partial }{%
\partial z_{i,k}^{\infty }}H^{\prime }\left( z_{i}^{\infty },\theta \right)
\left( \hat{\mathbb{E}}\left[ H\left( z^{\infty },\theta \right) H^{\prime
}\left( z^{\infty },\theta \right) \right] \right) ^{-1}\right] \hat{\mathbb{%
E}}\left[ H\left( z^{\infty },\theta \right) H^{\prime }\left( z^{\infty
},\theta \right) \right] \lambda \\
&&+H^{\prime }\left( z_{i}^{\infty },\theta \right) \left( \hat{\mathbb{E}}%
\left[ H\left( z^{\infty },\theta \right) H^{\prime }\left( z^{\infty
},\theta \right) \right] \right) ^{-1}n^{-1}\left[ \frac{\partial }{\partial
z_{i,k}^{\infty }}H\left( z_{i}^{\infty },\theta \right) \right] H^{\prime
}\left( z_{i}^{\infty },\theta \right) \lambda
\end{eqnarray*}%
This expression (once all derivatives of products are expanded) has the
general form of a product of $\lambda $ with functions of $\mathbf{z}$ that
contain terms of the form $\left( \hat{\mathbb{E}}\left[ H\left( z,\theta
\right) H^{\prime }\left( z,\theta \right) \right] \right) ^{-1}$, which is
nonsingular for $\mathbf{z}\in \eta $ by Assumption \ref{assitconv}(ii), and
derivatives of $g\left( z,\theta \right) $ up to order 2, which are bounded
for $z$ in the compact set $\left\{ z_{j}:\left( z_{1},\ldots ,z_{n}\right)
\in \eta \text{ and }j=1,\dots ,n\right\} $ by Assumption \ref{assitconv}%
(i). Hence the elements $\partial z_{j}^{t+1}/\partial z_{i,k}^{t}$ are
bounded by a constant times $\lambda $. It follows that the eigenvalues of
the matrix of partial derivatives of $f\left( \mathbf{z}\right) $ can be
made strictly less than 1 for $\lambda $ sufficiently small.
\end{proof}

\section{Constrained estimator}

\label{secconstr}

In some applications, it is useful to constrain the 
error, for instance to enforce the known fact that some variables are
measured without error. The optimization problem then becomes to 
minimize $\hat{\mathbb{E}}\left[ \left\Vert z-x\right\Vert ^{2}\right] $
subject to 
\begin{eqnarray}
\hat{\mathbb{E}}\left[ g\left( z,\theta \right) \right] &=&0 \\
C\left( z_{i}-x_{i}\right) &=&0\text{ for }i=1,\ldots ,n  \label{eqconstr}
\end{eqnarray}%
for some known rectangular full row rank matrix $C$ that selects the
dimensions of the error vector $x_{i}-z_{i}$ that should be
constrained to be zero. Note that error constraint is imposed
for each observation, not in an average sense. Naturally, we assume that the number of constraints imposed is not so large that it is no longer possible to satisfy all moment conditions simultaneously.

\begin{proposition}
The first order conditions (\ref{eqfoctheta})
and (\ref{eqfoclambda}) are unchanged, while Equation (\ref{eqfoczj})
becomes:%
\begin{equation}
\left( z_{j}-x_{j}\right) -PH^{\prime }\left( x,\theta \right) \lambda =0
\label{eqfocQzj}
\end{equation}%
where $P=\left( I-C^{\prime }\left( CC^{\prime }\right) ^{-1}C\right) $ and $%
H^{\prime }\left( x,\theta \right) \equiv \partial_{z}g^{\prime }\left(
z_{j},\theta \right) $.
\end{proposition}

Thanks to linearity, the Lagrange multipliers of the error constraints
can be explicitly solved for and the dimensionality of the
problem is not increased. The only effect of the constraints is to
\textquotedblleft project out\textquotedblright , through the matrix $P$,\
the dimensions where there are no errors.

\begin{proof}
The Lagrangian for this problem is%
\begin{equation}
\frac{1}{2}\hat{\mathbb{E}}\left\Vert z-x\right\Vert ^{2}-\lambda ^{\prime }%
\hat{\mathbb{E}}\left[ g\left( z,\theta \right) \right] -\sum_{i=1}^{n}%
\gamma _{i}^{\prime }C\left( z_{i}-x_{i}\right)  \label{eqLagrancons}
\end{equation}%
where $\lambda $ and $\gamma _{i}$ are Lagrange multipliers. 
The first order conditions of the Lagrangian (\ref{eqLagrancons}) with
respect to $z_{j}$ is%
\begin{equation}
\left( z_{j}-x_{j}\right) -\partial_{z} g^{\prime }\left( z_{j},\theta
\right) \lambda -C^{\prime }\gamma _{j}=0.  \label{eqconsgam}
\end{equation}%
Re-arranging and pre-multiplying both sides by the full column rank matrix $%
C $ yields:%
\begin{equation*}
C\left( z_{j}-x_{j}\right) -CC^{\prime }\gamma _{j}=C\partial_{z} g^{\prime
}\left( z_{j},\theta \right) \lambda ,
\end{equation*}%
thus allowing us to solve for $\gamma _{j}$:%
\begin{equation*}
\gamma _{j}=-\left( CC^{\prime }\right) ^{-1}C\partial_{z} g^{\prime }\left(
z_{j},\theta \right) \lambda .
\end{equation*}%
Upon substitution of $\gamma _{i}$ into (\ref{eqconsgam}) and simple
re-arrangements, we obtain%
\begin{eqnarray*}
\left( z_{j}-x_{j}\right) &=&\left( I-C^{\prime }\left( CC^{\prime }\right)
^{-1}C\right) \partial_{z} g^{\prime }\left( z_{j},\theta \right) \lambda \\
&=&PH\left( z,\theta \right) \lambda
\end{eqnarray*}%
where $P=\left( I-C^{\prime }\left( CC^{\prime }\right) ^{-1}C\right) $ and $%
H\left( z,\theta \right) =\partial_{z^{\prime }}g\left( z,\theta \right) $.

\end{proof}

The iterative Algorithm \ref{algoiter} can easily be adapted by replacing
every instance of $H^{\prime }\left( z^{t},\theta \right) $ by $PH^{\prime
}\left( z^{t},\theta \right) $. Similarly the linearized estimator of
Equation (\ref{eqsmlerr}) becomes:%
\begin{equation*}
\frac{1}{2}\hat{\mathbb{E}}\left[ g^{\prime }\left( x,\theta \right) \right]
\left( \hat{\mathbb{E}}\left[ H\left( x,\theta \right) PH^{\prime }\left(
x,\theta \right) \right] \right) ^{-1}\hat{\mathbb{E}}\left[ g\left(
x,\theta \right) \right] .
\end{equation*}
Note that this expression assumes that the matrix being inverted remains full rank, a condition that can be interpreted as requiring the constraints to not be so strong as to make impossible to simultaneously satisfy all the moment conditions.

Asymptotic results can also be straightforwardly adapted.

\begin{corollary}
Theorem \ref{thnormsml} holds under constraint (\ref{eqconstr}), with all
instances of $\mathbb{E}[H_iH_i^{\prime }]$ replaced by $\mathbb{E}%
[H_i PH_i^{\prime }]$, for $P=\left( I-C^{\prime }\left( CC^{\prime
}\right) ^{-1}C\right) $.
\end{corollary}

\section{Asymptotics: Proofs of Lemmas}
\label{secprooflem}

\begin{proof}[Proof of Lemma \protect\ref{lemdiffq}]
Since $h\left( z,\theta ,\lambda \right) $ is assumed to be continuous in
all of its arguments, there only remains to show that $q\left( x,\theta
,\lambda \right) $ is continuous in $\left( \theta ,\lambda \right) $. In
fact, we can establish the stronger statement that $q\left( x,\theta
,\lambda \right) $ is differentiable in $\left( \theta ,\lambda \right) $.
Differentiability in $\theta $ can be shown by the implicit function theorem%
\begin{eqnarray*}
\partial_{\theta^{\prime }}q\left( x,\theta ,\lambda \right) &=&\left[ \left( 
\frac{\partial }{\partial z^{\prime }}\left( z-\partial_{z}\left( \lambda
^{\prime }g\left( z,\theta \right) \right) \right) \right) ^{-1}\frac{%
\partial }{\partial \theta ^{\prime }}\partial_{z}\left( z-\lambda ^{\prime
}g\left( z,\theta \right) \right) \right] _{z=q\left( x,\theta ,\lambda
\right) } \\
&=&\left[ \left( I-\partial_{zz^{\prime }}\left( \lambda ^{\prime }g\left(
z,\theta \right) \right) \right) ^{-1}\partial_{z\theta^{\prime }}\left( \lambda
^{\prime }g\left( z,\theta \right) \right) \right] _{z=q\left( x,\theta
,\lambda \right) }
\end{eqnarray*}%
since $q\left( x,\theta ,\lambda \right) $ is the inverse of $%
z\mapsto z-\partial_{z}\left( \lambda ^{\prime }g\left( z,\theta \right)
\right) $. By the definition of $\bar{\lambda}$, $\bar{\nu}$, 
\begin{equation*}
\left\Vert \left( I-\partial_{zz^{\prime }}\left( \lambda ^{\prime }g\left(
z,\theta \right) \right) \right) ^{-1}\partial_{z\theta^{\prime }}\left( \lambda
^{\prime }g\left( z,\theta \right) \right) \right\Vert \leq \left( 1-\bar{%
\lambda}\bar{\nu}\right) ^{-1}\left\Vert \partial_{z\theta^{\prime }}\left(
\lambda ^{\prime }g\left( z,\theta \right) \right) \right\Vert ,
\end{equation*}%
at $z=q\left( x,\theta ,\lambda \right) $, where $\bar{\lambda}\bar{\nu}<1$
by Assumption \ref{asseigval} and where $\partial_{z\theta^{\prime }}\left(
\lambda ^{\prime }g\left( z,\theta \right) \right) $ exists by Assumption %
\ref{assgcontdiff}. Thus $h\left( q\left( x,\theta ,\lambda \right) ,\theta
,\lambda \right) $ is continuous in $\theta $.

By a similar reasoning, we can show that $h\left( q\left( x,\theta ,\lambda
\right) ,\theta ,\lambda \right) $ is continuous in $\lambda $ if we can
show that $\partial_{\lambda^{\prime }}q\left( x,\theta ,\lambda \right) $
exists: 
\begin{eqnarray*}
\left\Vert \partial_{\lambda^{\prime }}q\left( x,\theta ,\lambda \right)
\right\Vert &=&\left\Vert \left[ \left( \frac{\partial }{\partial z^{\prime }%
}\left( z-\partial_{z}\left( \lambda ^{\prime }g\left( z,\theta \right)
\right) \right) \right) ^{-1}\frac{\partial }{\partial \lambda ^{\prime }}%
\left( \partial_{z}g^{\prime }\left( z,\theta \right) \lambda \right) %
\right] _{z=q\left( x,\theta ,\lambda \right) }\right\Vert \\
&\leq &\left( 1-\bar{\lambda}\bar{\nu}\right) ^{-1}\left\Vert \left[
\partial_{z}g^{\prime }\left( z,\theta \right) \right] _{z=q\left( x,\theta
,\lambda \right) }\right\Vert
\end{eqnarray*}%
where $\bar{\lambda}\bar{\nu}<1$ by Assumption \ref{asseigval} and $%
\partial_{z}g^{\prime }\left( z,\theta \right) $ exists by Assumption \ref%
{assgcontdiff}.
\end{proof}

\begin{proof}[Proof of Lemma \protect\ref{lemholderLD}]
By the triangle inequality and Definition \ref{defholderLD}, there exists $%
\bar{h}\left( x,\theta \right) $ such that:%
\begin{equation}
\left\Vert h\left( z,\theta \right) \right\Vert \leq \left\Vert h\left(
x,\theta \right) \right\Vert +\left\Vert h\left( z,\theta \right) -h\left(
x,\theta \right) \right\Vert  \leq \left\Vert h\left( x,\theta \right) \right\Vert +\bar{h}\left( x,\theta
\right) \left\Vert z-x\right\Vert ,  \label{eqmvH}
\end{equation}%
for $z=q\left( x,\theta ,\lambda \right) $.\ Next, using the first order
conditions (Equation (\ref{eqfoczj})), we have, by a mean value argument,
the triangle inequality and the definitions of $\bar{\lambda}$ and $\bar{\nu}
$ from Assumption \ref{asseigval}, for some mean value $\tilde{x}$, %
\begin{eqnarray}
\left\Vert z-x\right\Vert &=&\left\Vert \partial_{z}\left( \lambda ^{\prime
}g\left( z,\theta \right) \right) \right\Vert  \leq \left\Vert \partial_{z}\left( \lambda ^{\prime }g\left( x,\theta
\right) \right) \right\Vert +\left\Vert \partial_{zz^{\prime }}\left(
\lambda ^{\prime }g\left( \tilde{x},\theta \right) \right) \left( z-x\right)
\right\Vert \notag \\
&\leq &\bar{\lambda}\left\Vert \partial_{z^{\prime }}g\left( x,\theta
\right) \right\Vert +\bar{\lambda}\bar{\nu}\left\Vert z-x\right\Vert .
\end{eqnarray}%
Re-arranging and using the fact that $\bar{\lambda}\bar{\nu}<1$ by
Assumption \ref{asseigval} and $\bar{\lambda}<\infty $ by compactness of $%
\Lambda $,%
\begin{equation}
\left\Vert z-x\right\Vert \leq \frac{\bar{\lambda}\left\Vert \partial_{z^{\prime }}g\left( x,\theta \right) \right\Vert }{\left( 1-\bar{\lambda}%
\bar{\nu}\right) }.  \label{eqzxdiff}
\end{equation}%
Combining (\ref{eqmvH}) and (\ref{eqzxdiff}) and noting that applying the $%
\mathbb{E}\left[ \sup_{\theta \in \Theta }\ldots \right] $ operator does not
alter the inequalities, we have%
\begin{equation*}
\mathbb{E}\left[ \sup_{\theta \in \Theta }\left\Vert h\left( z,\theta
\right) \right\Vert \right] \leq \mathbb{E}\left[ \sup_{\theta \in \Theta
}\left\Vert h\left( x,\theta \right) \right\Vert \right] +\frac{\bar{\lambda}%
}{\left( 1-\bar{\lambda}\bar{\nu}\right) }\mathbb{E}\left[ \sup_{\theta \in
\Theta }\bar{h}\left( x,\theta \right) \left\Vert \partial_{z^{\prime}}g\left( x,\theta \right) \right\Vert \right]
\end{equation*}%
where the right-hand side quantities are finite by construction since $h\in 
\mathcal{L}$.
\end{proof}

\section{Application: Additional results}

\subsection{Regression with all controls}
\label{appallcon}

In this section, we estimate a regression with controls, which now also includes average distance to the nearest body of water; land gradient; dummy variables for census regions; log share of the fraction of adult population with a college degree or more; log average income per capita; log share of employment in wholesale trade; and log average daily traffic on the interstate highways in 2005. The results, reported in Table \ref{all_controls}, are similar to those of the main regression shown in Table \ref{main_results}. In Table \ref{all_control_sets}, we also show how the main coefficient of interest changes when including various subsets of controls.

\begin{table}[!ht]
\caption{Results with all controls}
\label{all_controls}%
\begin{equation*}
\begin{array}{|l|l|l|}
\hline
& \mbox{GMM} & \mbox{OTGMM} \\ \hline
\mbox{Dependent variable} & \multicolumn{2}{c}{\mbox{Exporter fixed effect weight}} \vline \\ \hline
\mbox{log highway km} & 0.45 & 0.47 \\ 
\mbox{se} & (0.16) & (0.17)\\ \hline
\mbox{log employment} & 0.67&  1.06\\
\mbox{se} & (0.42) & (1.12)\\ \hline
\mbox{market access (export)} & -0.49& -0.50\\
\mbox{se} & (0.10) & (0.10)\\ \hline
\mbox{log 1920 population} & -0.29& -0.40\\
\mbox{se} & (0.24) & (0.30)\\ \hline
\mbox{log 1950 population} &  0.71&  0.88\\
\mbox{se} & (0.39) & (0.50)\\ \hline
\mbox{log 2000 population} & -0.65& -1.17\\
\mbox{se} & (0.47) & (1.30)\\ \hline
\mbox{log \% manuf empl} &  0.57&  0.54\\
\mbox{se} & (0.14) & (0.14)\\ \hline
\mbox{Water} & 0.07 & 0.06\\
\mbox{se} & (0.05) & (0.05)\\ \hline
\mbox{Land gradient} & -0.21& -0.20\\
\mbox{se} & (0.08) & (0.09)\\ \hline
\mbox{College} & -0.56 & -0.77\\
\mbox{se} & (0.47)& (0.54)\\ \hline
\mbox{price}  & 0.34 & 0.59\\
\mbox{se} & (0.57)& (0.54)\\ \hline
\mbox{wholesale} & 0.74 & 0.73\\
\mbox{se} & (0.26) & (0.63)\\ \hline
\mbox{traffic} & 0.37 & 0.38\\
\mbox{se} & (0.30) & (0.29)\\ \hline
\end{array}%
\end{equation*}
{Results with all controls. Replicated estimates from the original paper and OTGMM estimates. Heteroskedasticity-robust standard errors (GMM) and small-error asymptotic standard error (OTGMM) in parentheses. The regression also include an intercept and census region dummies.}
\end{table}

\begin{table}[!ht]
\caption{Coefficient on log highway kilometers for various control sets}
\label{all_control_sets}%
\begin{equation*}
\begin{array}{|l|l|l|}
\hline
\mbox{Controls} & \mbox{GMM} & \mbox{OTGMM}\\ \hline
\mbox{none} & 1.13 & 1.17\\ 
 & (0.14) & (0.15)\\ \hline
\mbox{2} & 0.57 & 0.65\\ 
 & (0.16) & (0.19)\\ \hline
\mbox{1-5} & 0.47 & 0.46\\ 
 & (0.13) & (0.12)\\ \hline
\mbox{1-6} & 0.39 & 0.40\\ 
 & (0.12) & (0.11)\\ \hline
\mbox{1-7} & 0.34 & 0.37\\ 
 & (0.15) & (0.13)\\ \hline
\mbox{1-8} & 0.28 & 0.31\\ 
 & (0.14) & (0.13)\\ \hline
\mbox{1-9} & 0.27 & 0.30\\ 
 & (0.14) & (0.13)\\ \hline
\mbox{1-10} & 0.27 & 0.29\\ 
 & (0.13) & (0.13)\\ \hline
 \mbox{1-11} & 0.38 & 0.45\\ 
 & (0.13) & (0.14)\\ \hline
\mbox{1-12} & 0.44 & 0.48\\ 
 & (0.16) & (0.17)\\ \hline
\mbox{All} & 0.45 & 0.47\\ 
 & (0.16) & (0.17)\\ \hline
\end{array}%
\end{equation*}
{Coefficient on log highway kilometers for various sets of controls, numbered as follows: employment (1), market access (2), population 1920 (3), population 1950 (4), population 2000 (5), log \% manuf employment (6), water (7), land gradient (8), college (9), price (10), wholesale (11), traffic (12), region census dummies (13-20). All regressions include an intercept.}
\end{table}

\subsection{Relaxing exclusion restrictions}
\label{apprelaxex}

As \citet{poirier:salvaging}, we consider the possibility that some instrumental variables have non-zero regression coefficient, i.e., have a direct impact on the outcome. 
Since our estimator relies on over-identifying restrictions to recover meaningful variable corrections, we only relax one of the exclusion-restrictions at a time. This leads to the following results. 

\begin{table}[!ht]
\caption{Including log 1947 highway kilometers as a regressor}
\begin{equation*}
\begin{array}{|l|l|l|}
\hline
& \mbox{GMM} & \mbox{OTGMM} \\ \hline
\mbox{Dependent variable} & \multicolumn{2}{c}{\mbox{Exporter fixed effect weight}} \vline \\ \hline
\mbox{log highway km}  & 0.52 & 0.46 \\ 
\mbox{se} & (0.51) & (0.52) \\ \hline
\mbox{log 1947 highway km}  & -0.09 & -0.04 \\ 
\mbox{se} & (0.36) & (0.36) \\ \hline
\mbox{log employment}  & 0.44 & 1.24 \\ 
\mbox{se} & (0.33) & (0.37) \\ \hline
\mbox{market access (export)}  & -0.63 & -0.65  \\ 
\mbox{se} & (0.11) & (0.10) \\ \hline
\mbox{log 1920 population}  & -0.30 & -0.58  \\ 
\mbox{se} & (0.24) & (0.23) \\ \hline
\mbox{log 1950 population}  & 0.64 & 1.15  \\ 
\mbox{se} & (0.38) & (0.38) \\ \hline
\mbox{log 2000 population} & -0.19 & -1.26  \\ 
\mbox{se} & (0.46) & (0.38) \\ \hline
\mbox{log \% manuf empl}  & 0.64 & 0.58  \\ 
\mbox{se} & (0.12) & (0.12) \\ \hline
\end{array}%
\end{equation*}
{Including log 1947 highway kilometers as a regressor. Replicated estimates from the original paper and OTGMM estimates. Heteroskedasticity-robust standard errors (GMM) and small-error asymptotic standard error (OTGMM) in parentheses}
\end{table}

\begin{table}[!ht]
\caption{Including log exploration as a regressor}
\begin{equation*}
\begin{array}{|l|l|l|}
\hline
& \mbox{GMM} & \mbox{OTGMM} \\ \hline
\mbox{Dependent variable} & \multicolumn{2}{c}{\mbox{Exporter fixed effect weight}} \vline \\ \hline
\mbox{log highway km} & 0.42 & 0.44 \\ 
\mbox{se} & (0.16) & (0.16) \\ \hline
\mbox{log exploration} & -0.02 & -0.03 \\ 
\mbox{se} & (0.06) & (0.06) \\ \hline
\mbox{log employment} & 0.47 & 1.18 \\ 
\mbox{se} & (0.34) & (0.36) \\ \hline
\mbox{market access (export)} & -0.63 & -0.65 \\ 
\mbox{se} & (0.11) & (0.11) \\ \hline
\mbox{log 1920 population} & -0.29 & -0.52  \\ 
\mbox{se} & (0.23) & (0.26) \\ \hline
\mbox{log 1950 population} & 0.64 & 1.05 \\ 
\mbox{se} & (0.38) & (0.43) \\ \hline
\mbox{log 2000 population} & -0.20 & -1.16 \\ 
\mbox{se} & (0.45) & (0.43) \\ \hline
\mbox{log \% manuf empl} & 0.63 & 0.57 \\ 
\mbox{se} & (0.13) & (0.12) \\ \hline
\end{array}%
\end{equation*}
{Including log exploration as a regressor. Replicated estimates from the original paper and OTGMM estimates. Heteroskedasticity-robust standard errors (GMM) and small-error asymptotic standard error (OTGMM) in parentheses}
\end{table}

\begin{table}[!ht]
\caption{Including log 1998 railroad as a regressor}
\begin{equation*}
\begin{array}{|l|l|l|}
\hline
& \mbox{GMM} & \mbox{OTGMM} \\ \hline
\mbox{Dependent variable} & \multicolumn{2}{c}{\mbox{Exporter fixed effect weight}} \vline \\ \hline
\mbox{log highway km} & 0.23 & 0.24 \\ 
\mbox{se} & (0.14) & (0.14) \\ \hline
\mbox{log 1998 railroad} & 0.16 & 0.15 \\ 
\mbox{se} & (0.10) & (0.10) \\ \hline
\mbox{log employment} & 0.44 & 0.52 \\ 
\mbox{se} & (0.31) & (0.44) \\ \hline
\mbox{market access (export)} & -0.63 & -0.63 \\ 
\mbox{se} & (0.11) & (0.11) \\ \hline
\mbox{log 1920 population} & -0.41 & -0.41 \\ 
\mbox{se} & (0.24) & (0.30) \\ \hline
\mbox{log 1950 population} & 0.76 & 0.77 \\ 
\mbox{se} & (0.37) & (0.48) \\ \hline
\mbox{log 2000 population} & -0.16 & -0.26 \\ 
\mbox{se} & (0.41) & (0.48) \\ \hline
\mbox{log \% manuf empl} & 0.62 & 0.61\\ 
\mbox{se} & (0.12) & (0.56) \\ \hline
\end{array}%
\end{equation*}
{Including log 1998 railroad as a regressor. Replicated estimates from the original paper and OTGMM estimates. Heteroskedasticity-robust standard errors (GMM) and small-error asymptotic standard error (OTGMM) in parentheses}
\end{table}

As in \citet{poirier:salvaging}, we found that removing railroads has the largest impact on coefficients and is likely an improper instrument.  

\clearpage

\section{Simulations}
\label{secsimul}

We conduct simulations to assess the performance of our estimator and
compare it to efficient GMM. We consider both the OTGMM estimator (Equations
(\ref{eqobjf}) and (\ref{eqcons})) and the GMM estimator obtained under the assumption of small errors (Equation (\ref{eqsmlerr})).

Given the nonlinear nature of our estimator, we deliberately select a small sample size ($n=100$) to explore the regime where asymptotic results do not trivially hold.
We consider various moment conditions,
underlying distributions and signal-to-noise ratios. There is an underlying
random variable $z_{i}$ which satisfies the moment conditions, but the
researcher observes $x_{i}=z_{i}+\sigma e_{i}$ with $e_{i}\sim \mathcal{N}%
(0,1)$. We consider different values for the error scale $\sigma $
in order to assess the impact of magnitude of the errors on the
performance of estimators that only use the observed $x_{i}$.

Specifically, we consider the following distributions for $z_{i}$: $%
z_{i}\sim \mathcal{N}(1.5,2)$ $z_{i}\sim \mbox{Unif}[1,2]$ (uniform) $%
z_{i}\sim \mathcal{B}(5,0.3)$ (binomial) and $\sigma =0,0.5,1,1.5,2,2.5$.
The true parameter value is $\theta _{0}=1.5$, as obtained by the following
moment conditions:%
\begin{eqnarray}
\mathbb{E}[z_{i}-\theta ] &=&0,~\mathbb{E}\left[e^{z_{i}}-\frac{2}{3}\theta 
\mathbb{E}[e^{z_{i}}]\right]=0  \label{eqsimmom1} \\
\mathbb{E}[z_{i}-\theta ] &=&0,~\mathbb{E}\left[\frac{e^{2z_{i}-3}}{1+e^{2z-3}}-%
\frac{2}{3}\theta \frac{e^{2z_{i}-3}}{1+e^{2z-3}}\right]=0  \label{eqsimmom2} \\
\mathbb{E}[e^{z_{i}}-\frac{2}{3}\theta \mathbb{E}[e^{z_{i}}]] &=&0,~\mathbb{E%
}\left[\frac{e^{2z_{i}-3}}{1+e^{2z-3}}-\frac{\theta e^{2z_{i}-3}}{\left(
1.5\right) \left( 1+e^{2z-3}\right) }\right]=0  \label{eqsimmom3}
\end{eqnarray}%
Finally, we consider the process: $z_{i}\sim \mbox{Exp}(\frac{2}{3})$
with the moment conditions 
\begin{equation}
\mathbb{E}[z_{i}-\theta ]=0,~\mathbb{E}[z{_{i}}^{2}-2\theta ^{2}]=0.
\label{eqsimmom4}
\end{equation}%
In all cases, the model is correctly specified in the absence of 
errors ($\sigma =0$) but starts to violate the overidentifying restrictions when
there are errors ($\sigma >0$ so that $x\not=z$).

\begin{table}[ht]
\caption{Simulation results: Equation (\protect\ref{eqsimmom1})}
\label{tab1}\centering        
\resizebox{\columnwidth}{!}{\begin{tabular}{| l | l | l | l | l | l | l | l | l | l | l | l | l | l | l | l | l | l | l |}
\hline
            & \multicolumn{18}{c|}{\mbox{Bias}} \\
\hline
            & \multicolumn{6}{|c}{\mbox{Linearized OTGMM}} & \multicolumn{6}{|c}{\mbox{OTGMM}} & \multicolumn{6}{|c|}{\mbox{Efficient GMM}} \\
\hline
error scale & 0 & 0.5   & 1     & 1.5   & 2      & 2.5    & 0     & 0.5   & 1     & 1.5   & 2     & 2.5   & 0             & 0.5   & 1     & 1.5   & 2     & 2.5   \\
Normal & 0.00 & -0.01 & -0.05 & -0.16 & -0.49 & -1.61 & 0.00 & -0.01 & -0.03 & -0.04 & -0.05 & -0.04 & -0.02 & -0.03 & -0.09 & -0.17 & -0.23 & -0.25\\
Uniform & 0.00 & -0.23 & -1.12 & -3.59 & -10.93 & -36.42 & 0.00 & -0.09 & -0.12 & -0.11 & -0.08 & -0.06 & 0.00 & -0.13 & -0.22 & -0.27 & -0.29 & -0.29\\
Binomial & 0.00 & -0.02 & -0.10 & -0.32 & -0.99 & -3.28 & 0.00 & -0.01 & -0.04 & -0.06 & -0.06 & -0.05 & 0.00 & -0.04 & -0.14 & -0.21 & -0.25 & -0.26\\

\hline
            & \multicolumn{18}{c|}{\mbox{Standard deviation}} \\
\hline
            & \multicolumn{6}{|c}{\mbox{Linearized OTGMM}} & \multicolumn{6}{|c}{\mbox{OTGMM}} & \multicolumn{6}{|c|}{\mbox{Efficient GMM}}\\
\hline
error scale & 0                    & 0.5   & 1     & 1.5   & 2      & 2.5    & 0     & 0.5   & 1     & 1.5   & 2     & 2.5   & 0             & 0.5   & 1     & 1.5   & 2     & 2.5   \\
Normal & 0.14 & 0.15 & 0.16 & 0.22 & 0.76 & 4.43 & 0.14 & 0.15 & 0.17 & 0.21 & 0.25 & 0.29 & 0.16 & 0.16 & 0.16 & 0.19 & 0.24 & 0.29\\
Uniform & 0.01 & 0.05 & 0.27 & 1.41 & 8.09 & 57.82 & 0.01 & 0.04 & 0.10 & 0.15 & 0.20 & 0.25 & 0.01 & 0.04 & 0.09 & 0.14 & 0.20 & 0.26\\
Binomial & 0.10 & 0.10 & 0.13 & 0.25 & 1.18 & 7.23 & 0.10 & 0.11 & 0.14 & 0.18 & 0.23 & 0.27 & 0.10 & 0.10 & 0.13 & 0.17 & 0.22 & 0.28\\

\hline
            & \multicolumn{18}{c|}{\mbox{RMSE}} \\
\hline
            & \multicolumn{6}{|c}{\mbox{Linearized OTGMM}} & \multicolumn{6}{|c}{\mbox{OTGMM}} & \multicolumn{6}{|c|}{\mbox{Efficient GMM}} \\
\hline
error scale  & 0                    & 0.5   & 1     & 1.5   & 2      & 2.5    & 0     & 0.5   & 1     & 1.5   & 2     & 2.5   & 0             & 0.5   & 1     & 1.5   & 2     & 2.5   \\
Normal & 0.14 & 0.15 & 0.17 & 0.27 & 0.90 & 4.71 & 0.14 & 0.15 & 0.17 & 0.21 & 0.25 & 0.29 & 0.17 & 0.16 & 0.19 & 0.26 & 0.33 & 0.39\\
Uniform & 0.01 & 0.24 & 1.15 & 3.86 & 13.60 & 68.33 & 0.01 & 0.10 & 0.15 & 0.18 & 0.22 & 0.26 & 0.01 & 0.14 & 0.24 & 0.31 & 0.35 & 0.39\\
Binomial & 0.10 & 0.11 & 0.16 & 0.41 & 1.54 & 7.94 & 0.10 & 0.11 & 0.14 & 0.19 & 0.23 & 0.28 & 0.10 & 0.11 & 0.19 & 0.27 & 0.33 & 0.38\\

\hline
\end{tabular}}
\end{table}

\begin{table}[ht]
\caption{Simulation results: Equation (\protect\ref{eqsimmom2})}
\label{tab2}\centering        
\resizebox{\columnwidth}{!}{\begin{tabular}{| l | l | l | l | l | l | l | l | l | l | l | l | l | l | l | l | l | l | l |}
\hline
            & \multicolumn{18}{c|}{\mbox{Bias}} \\
\hline
            & \multicolumn{6}{|c}{\mbox{Linearized OTGMM}} & \multicolumn{6}{|c}{\mbox{OTGMM}} & \multicolumn{6}{|c|}{\mbox{Efficient GMM}} \\
\hline
error scale & 0                    & 0.5   & 1     & 1.5   & 2      & 2.5    & 0     & 0.5   & 1     & 1.5   & 2     & 2.5   & 0             & 0.5   & 1     & 1.5   & 2     & 2.5   \\
Normal & 0.00 & 0.00 & 0.00 & 0.00 & 0.00 & 0.00 & 0.00 & 0.00 & 0.00 & 0.00 & 0.00 & 0.00 & 0.00 & 0.00 & 0.00 & 0.00 & 0.00 & 0.00\\
Uniform & 0.00 & 0.00 & 0.00 & 0.00 & 0.00 & 0.00 & 0.00 & 0.00 & 0.00 & 0.00 & 0.00 & 0.00 & 0.00 & 0.00 & 0.00 & 0.00 & 0.00 & 0.00\\
Binomial & 0.00 & 0.00 & 0.01 & 0.01 & 0.02 & 0.02 & 0.00 & 0.01 & 0.02 & 0.03 & 0.03 & 0.04 & -0.01 & 0.01 & 0.04 & 0.05 & 0.06 & 0.05\\

\hline
            & \multicolumn{18}{c|}{\mbox{Standard deviation}} \\
\hline
            & \multicolumn{6}{|c}{\mbox{Linearized OTGMM}} & \multicolumn{6}{|c}{\mbox{OTGMM}} & \multicolumn{6}{|c|}{\mbox{Efficient GMM}}\\
\hline
error scale & 0                    & 0.5   & 1     & 1.5   & 2      & 2.5    & 0     & 0.5   & 1     & 1.5   & 2     & 2.5   & 0             & 0.5   & 1     & 1.5   & 2     & 2.5   \\
Normal & 0.11 & 0.12 & 0.12 & 0.13 & 0.15 & 0.16 & 0.11 & 0.11 & 0.12 & 0.12 & 0.13 & 0.13 & 0.10 & 0.10 & 0.09 & 0.09 & 0.09 & 0.09\\
Uniform & 0.00 & 0.04 & 0.15 & 0.29 & 0.45 & 0.61 & 0.00 & 0.05 & 0.10 & 0.12 & 0.12 & 0.13 & 0.00 & 0.04 & 0.10 & 0.10 & 0.09 & 0.09\\
Binomial & 0.10 & 0.11 & 0.13 & 0.15 & 0.18 & 0.20 & 0.10 & 0.11 & 0.12 & 0.13 & 0.13 & 0.14 & 0.10 & 0.10 & 0.10 & 0.10 & 0.10 & 0.10\\
\hline
            & \multicolumn{18}{c|}{\mbox{RMSE}} \\
\hline
            & \multicolumn{6}{|c}{\mbox{Linearized OTGMM}} & \multicolumn{6}{|c}{\mbox{OTGMM}} & \multicolumn{6}{|c|}{\mbox{Efficient GMM}} \\
\hline
error scale & 0                    & 0.5   & 1     & 1.5   & 2      & 2.5    & 0     & 0.5   & 1     & 1.5   & 2     & 2.5   & 0             & 0.5   & 1     & 1.5   & 2     & 2.5   \\
Normal & 0.11 & 0.12 & 0.12 & 0.13 & 0.15 & 0.16 & 0.11 & 0.11 & 0.12 & 0.12 & 0.13 & 0.13 & 0.10 & 0.10 & 0.09 & 0.09 & 0.09 & 0.09\\
Uniform & 0.00 & 0.04 & 0.15 & 0.29 & 0.45 & 0.61 & 0.00 & 0.05 & 0.10 & 0.12 & 0.12 & 0.13 & 0.00 & 0.04 & 0.10 & 0.10 & 0.09 & 0.09\\
Binomial & 0.10 & 0.11 & 0.13 & 0.15 & 0.18 & 0.21 & 0.10 & 0.11 & 0.12 & 0.13 & 0.14 & 0.14 & 0.10 & 0.10 & 0.11 & 0.11 & 0.11 & 0.11\\
\hline
\end{tabular}}
\end{table}

\begin{table}[ht]
\caption{Simulation results: Equation (\protect\ref{eqsimmom3})}
\label{tab3}\centering        
\resizebox{\columnwidth}{!}{\begin{tabular}{| l | l | l | l | l | l | l | l | l | l | l | l | l | l | l | l | l | l | l |}
\hline
            & \multicolumn{18}{c|}{\mbox{Bias}} \\
\hline
            & \multicolumn{6}{|c}{\mbox{Linearized OTGMM}} & \multicolumn{6}{|c}{\mbox{OTGMM}} & \multicolumn{6}{|c|}{\mbox{Efficient GMM}} \\
\hline
error scale & 0                    & 0.5   & 1     & 1.5   & 2      & 2.5    & 0     & 0.5   & 1     & 1.5   & 2     & 2.5   & 0             & 0.5   & 1     & 1.5   & 2     & 2.5   \\
Normal & -0.01 & 0.00 & 0.02 & 0.09 & 0.29 & 1.00 & -0.01 & 0.00 & 0.00 & 0.00 & 0.00 & 0.00 & -0.01 & -0.01 & -0.04 & -0.07 & -0.07 & -0.07\\
Uniform & 0.00 & -0.13 & -0.63 & -2.02 & -6.17 & -20.84 & 0.00 & -0.05 & -0.04 & -0.02 & -0.01 & 0.00 & 0.00 & -0.12 & -0.16 & -0.15 & -0.12 & -0.09\\
Binomial & 0.00 & 0.01 & 0.04 & 0.09 & 0.24 & 0.76 & 0.00 & 0.01 & 0.02 & 0.03 & 0.03 & 0.04 & 0.00 & -0.01 & -0.05 & -0.06 & -0.06 & -0.04\\
\hline
            & \multicolumn{18}{c|}{\mbox{Standard deviation}} \\
\hline
            & \multicolumn{6}{|c}{\mbox{Linearized OTGMM}} & \multicolumn{6}{|c}{\mbox{OTGMM}} & \multicolumn{6}{|c|}{\mbox{Efficient GMM}}\\
\hline
error scale & 0                    & 0.5   & 1     & 1.5   & 2      & 2.5    & 0     & 0.5   & 1     & 1.5   & 2     & 2.5   & 0             & 0.5   & 1     & 1.5   & 2     & 2.5   \\
Normal & 0.11 & 0.11 & 0.12 & 0.15 & 0.49 & 2.77 & 0.11 & 0.11 & 0.11 & 0.12 & 0.12 & 0.13 & 0.12 & 0.11 & 0.11 & 0.12 & 0.13 & 0.13\\
Uniform & 0.04 & 0.06 & 0.27 & 1.13 & 5.89 & 44.86 & 0.04 & 0.06 & 0.09 & 0.11 & 0.12 & 0.12 & 0.04 & 0.06 & 0.09 & 0.11 & 0.12 & 0.13\\
Binomial & 0.10 & 0.10 & 0.11 & 0.15 & 0.72 & 6.05 & 0.10 & 0.10 & 0.11 & 0.12 & 0.12 & 0.13 & 0.10 & 0.10 & 0.11 & 0.12 & 0.13 & 0.14\\
\hline
            & \multicolumn{18}{c|}{\mbox{RMSE}} \\
\hline
            & \multicolumn{6}{|c}{\mbox{Linearized OTGMM}} & \multicolumn{6}{|c}{\mbox{OTGMM}} & \multicolumn{6}{|c|}{\mbox{Efficient GMM}} \\
\hline
error scale & 0                    & 0.5   & 1     & 1.5   & 2      & 2.5    & 0     & 0.5   & 1     & 1.5   & 2     & 2.5   & 0             & 0.5   & 1     & 1.5   & 2     & 2.5   \\
Normal & 0.11 & 0.11 & 0.12 & 0.18 & 0.58 & 2.95 & 0.11 & 0.11 & 0.11 & 0.12 & 0.12 & 0.13 & 0.12 & 0.11 & 0.12 & 0.14 & 0.15 & 0.15\\
Uniform & 0.04 & 0.14 & 0.69 & 2.31 & 8.53 & 49.47 & 0.04 & 0.08 & 0.10 & 0.11 & 0.12 & 0.12 & 0.04 & 0.13 & 0.19 & 0.19 & 0.17 & 0.16\\
Binomial & 0.10 & 0.10 & 0.11 & 0.17 & 0.76 & 6.09 & 0.10 & 0.10 & 0.11 & 0.12 & 0.13 & 0.13 & 0.10 & 0.10 & 0.12 & 0.14 & 0.14 & 0.14\\
\hline
\end{tabular}}
\end{table}

\begin{table}[ht]
\caption{Simulation results: Equation (\protect\ref{eqsimmom4})}
\label{tab4}\centering        
\resizebox{\columnwidth}{!}{\begin{tabular}{| l | l | l | l | l | l | l | l | l | l | l | l | l | l | l | l | l | l | l |}
\hline
            & \multicolumn{18}{c|}{\mbox{Bias}} \\
\hline
            & \multicolumn{6}{|c}{\mbox{Linearized OTGMM}} & \multicolumn{6}{|c}{\mbox{OTGMM}} & \multicolumn{6}{|c|}{\mbox{Efficient GMM}} \\
\hline
error scale & 0                    & 0.5   & 1     & 1.5   & 2      & 2.5    & 0     & 0.5   & 1     & 1.5   & 2     & 2.5   & 0             & 0.5   & 1     & 1.5   & 2     & 2.5   \\
Exponential & -0.01 & 0.03 & 0.16 & 0.38 & 0.67 & 0.99 & -0.01 & 0.03 & 0.13 & 0.27 & 0.42 & 0.58 & -0.04 & -0.02 & 0.10 & 0.29 & 0.50 & 0.71\\
\hline
            & \multicolumn{18}{c|}{\mbox{Standard deviation}} \\
\hline
            & \multicolumn{6}{|c}{\mbox{Linearized OTGMM}} & \multicolumn{6}{|c}{\mbox{OTGMM}} & \multicolumn{6}{|c|}{\mbox{Efficient GMM}}\\
\hline
error scale & 0                    & 0.5   & 1     & 1.5   & 2      & 2.5    & 0     & 0.5   & 1     & 1.5   & 2     & 2.5   & 0             & 0.5   & 1     & 1.5   & 2     & 2.5   \\
Exponential & 0.17 & 0.16 & 0.16 & 0.16 & 0.17 & 0.18 & 0.17 & 0.16 & 0.16 & 0.17 & 0.20 & 0.23 & 0.16 & 0.16 & 0.17 & 0.20 & 0.23 & 0.27\\
\hline
            & \multicolumn{18}{c|}{\mbox{RMSE}} \\
\hline
            & \multicolumn{6}{|c}{\mbox{Linearized OTGMM}} & \multicolumn{6}{|c}{\mbox{OTGMM}} & \multicolumn{6}{|c|}{\mbox{Efficient GMM}} \\
\hline
error scale & 0                    & 0.5   & 1     & 1.5   & 2      & 2.5    & 0     & 0.5   & 1     & 1.5   & 2     & 2.5   & 0             & 0.5   & 1     & 1.5   & 2     & 2.5   \\
Exponential & 0.17 & 0.17 & 0.23 & 0.42 & 0.69 & 1.00 & 0.17 & 0.17 & 0.21 & 0.32 & 0.47 & 0.62 & 0.16 & 0.16 & 0.20 & 0.35 & 0.55 & 0.76\\
\hline
\end{tabular}}
\end{table}

In Tables \ref{tab1}-\ref{tab4}, we report the estimation error $\hat{\theta}%
-\theta _{0}$ and decompose it into its bias, standard deviation and the
root mean square error (RMSE). These quantities are evaluated using
averages over 5000 replications. We consider various estimators $\hat{\theta}$:
the linear approximation to OTGMM in the small-error limit (leftmost
columns, indicated by ``Linearized OTGMM''), OTGMM in the general large-error case (middle columns) and
efficient GMM ignoring the presence of errors (rightmost columns).

While these tables consider a large number of specifications and data generating processes, we here highlight a few values that illustrate the typical trends present throughout the simulations.
First, it is clear that OTGMM is, in general, preferable to linearized OTGMM in terms of bias, and hence we focus our discussion on the former. As an example, let us compare the OTGMM and GMM results in Table \ref{tab1}, for the column that corresponds to an error scale of $\sigma=1.5$ for the normal model. The OTGMM estimator reduces the bias magnitude to only 0.05, as compared to 0.18 for GMM. At the same time, the variance only increases from 0.19 (for GMM) to 0.20 (for OTGMM), which is a negligible increase, thus supporting the idea that GMM's use of the optimal weighting matrix is not particularly beneficial in this context. As a result, the overall root mean square error (RMSE) is 0.21 for OTGMM down from 0.26 for GMM.

The key take-away from these simulations is that the OTGMM
estimator exhibits the ability to substantially reduce bias while not
substantially increasing the variance relative to efficient GMM. As a
result, the overall RMSE criterion points in favor of OTGMM. This is exactly
the type of behavior one would expect for an effective measurement
error-correcting method. The reduction in bias is especially important for
inference and testing, as it significantly reduces size distortion. In
contrast, a small increase in variance does not affect inference validity,
as this variance can be straightforwardly accounted for in the asymptotics,
unlike the bias, which is generally unknown.

\end{document}